\documentclass{article}

\usepackage{arxiv}

\usepackage[utf8]{inputenc} 
\usepackage[T1]{fontenc}    
\usepackage{hyperref}       
\usepackage{url}            
\usepackage{booktabs}       
\usepackage{amsfonts}       
\usepackage{nicefrac}       
\usepackage{microtype}      
\usepackage{lipsum}		
\usepackage{graphicx}
\usepackage[round]{natbib}
\usepackage{doi}

\usepackage{amsmath,amssymb,amsfonts}
\usepackage{algorithmic}
\usepackage{graphicx}
\usepackage{textcomp}
\usepackage{xcolor}

\usepackage{multicol}
\usepackage{multirow}
\usepackage{booktabs}
\usepackage[percent]{overpic}
\usepackage{rotating}

\allowdisplaybreaks 

\def\aujour{\number\day \space \ifcase\month\or
janvier\or f�vrier\or mars\or avril\or mai\or
juin\or juillet\or ao�t\or septembre\or octobre\or
novembre\or d�cembre\fi \space \number\year}

\def\cH{{\cal H}}

\def\cL{{\cal L}}

\newtheorem{remark}{Remark}

\newtheorem{theorem}{Theorem}

\newtheorem{proof}{Proof}
\def\C{{\setbox0=\hbox{$\displaystyle{\rm C}$}
        \hbox{\hbox to0pt{\kern 0.4\wd0\vrule height 0.95\ht0\hss}\box0}}}
\def\Q{{\setbox0=\hbox{$\displaystyle{\rm Q}$}%
    \hbox{\raise 0.2\ht0\hbox to0pt{\kern 0.4\wd0\vrule height
    0.85\ht0\hss}\box0}}} 
\def\R{\mathop{\rm I\mkern -3.5mu R}} 

\def\cHi{{\cal H}_{\infty}} 
\def\cH2{{\cal H}_2} 

\def\cL2{\mathop{\mathcal L}_{2}} 

\def\cRH2{\mathop{\cal R \cal H}_2} 
\def\cRL2{\mathop{\cal R \cal L}_{2}} 

\DeclareMathOperator*{\diag}{diag}

\DeclareMathOperator*{\der}{d}

\newcommand{\abs}[1]{\left|{#1}\right|}

\makeatletter
\DeclareRobustCommand\sfrac[1]{\@ifnextchar/{\@sfrac{#1}}
                                            {\@sfrac{#1}/}}
\def\@sfrac#1/#2{\leavevmode\kern.1em\raise.5ex
         \hbox{$\m@th\fontsize\sf@size\z@
                           \selectfont#1$}\kern-.1em
         /\kern-.15em\lower.25ex
          \hbox{$\m@th\fontsize\sf@size\z@
                            \selectfont#2$}}
\makeatother

\usepackage{cite}
\usepackage{siunitx}

\newcommand{\maxd}{d_{\mathcal{C}}}
\newcommand{\mind}{d_{\mathcal{B}}}

\newcommand{\ib}{I_\mathrm{b}}

\newcommand{\bitrate}{B}
\newcommand{\ber}{\mathrm{BER}}

\def\BibTeX{{\rm B\kern-.05em{\sc i\kern-.025em b}\kern-.08em
    T\kern-.1667em\lower.7ex\hbox{E}\kern-.125emX}}

\title{Establishing and Maintaining a Reliable Optical Wireless Communication in Underwater Environment}

\author{ \href{https://orcid.org/0000-0002-2576-5515}{\includegraphics[scale=0.06]{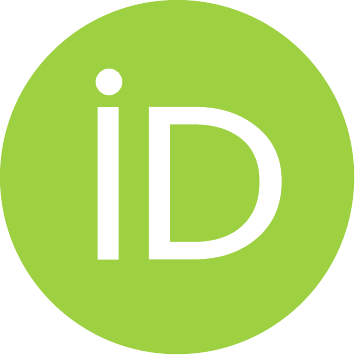}\hspace{1mm}Ibrahima~N'Doye$^1$}, 
\href{https://orcid.org/0000-0003-1819-9436}{\includegraphics[scale=0.06]{orcid.pdf}\hspace{1mm}Ding Zhang$^2$}, 
\href{https://orcid.org/0000-0003-4827-1793}{\includegraphics[scale=0.06]{orcid.pdf}\hspace{1mm}Mohamed-Slim Alouini$^1$}, 
\href{https://orcid.org/0000-0001-5944-0121}{\includegraphics[scale=0.06]{orcid.pdf}\hspace{1mm}Taous-Meriem~Laleg-Kirati$^1$}
\thanks{This work has been supported by the King Abdullah University of Science and Technology (KAUST), Base Research Fund (BAS/1/1627-01-01) to Taous Meriem Laleg.} \\
$^1$Computer, Electrical and Mathematical Sciences and Engineering Division (CEMSE)\\
King Abdullah University of Science and Technology (KAUST)\\
Thuwal 23955-6900, Saudi Arabia \\
	\texttt{ibrahima.ndoye@kaust.edu.sa; slim.alouini@kaust.edu.sa; taousmeriem.laleg@kaust.edu.sa} \\
$^2$Department of Electronic and Computer Engineering\\ 
The Hong Kong University of Science and Technology\\
 Clear Water Bay,  Kowloon, Hong Kong, China\\
\texttt{ding.zhang@connect.ust.hk}\\
}

\renewcommand{\shorttitle}{\textit{arXiv} Template}

\begin{document}
\maketitle

\begin{abstract}
This paper proposes the trajectory tracking problem between an autonomous underwater vehicle (AUV) and a mobile surface ship, both equipped with optical communication transceivers. The challenging issue is to maintain stable connectivity between the two autonomous vehicles within an optical communication range. We define a directed optical line-of-sight (LoS) link between the two-vehicle systems. The transmitter is mounted on the AUV while the surface ship is equipped with an optical receiver. However, this optical communication channel needs to preserve a stable transmitter-receiver position to reinforce service quality, which typically includes a bit rate and bit error rates. A cone-shaped beam region of the optical receiver is approximated based on the channel model; then, a minimum bit rate is ensured if the AUV transmitter remains inside of this region. Additionally, we design two control algorithms for the transmitter to drive the AUV and maintain it in the cone-shaped beam region under an uncertain oceanic environment. Lyapunov function-based analysis that ensures asymptotic stability of the resulting closed-loop tracking error is used to design the proposed NLPD controller. Numerical simulations are performed using MATLAB/Simulink to show the controllers' ability to achieve favorable tracking in the presence of the solar background noise within competitive times. Finally, results demonstrate the proposed NLPD controller improves the tracking error performance more than $70\%$ under nominal conditions and $35\%$ with model uncertainties and disturbances compared to the original PD strategy.
\end{abstract}


\section{Introduction}
Optical communication combines techniques from various long-studied disciplines: optical communication, free-space optical, and underwater optical communication, along with laser development and mathematical modeling \citep{Cox:12}; and its application draws more and more attention from the industry. The use of underwater wireless networks and the need for data for multimedia and other services with the aid of remotely-operated vehicles or autonomous underwater vehicles (AUVs) is a considered example \citep{HKBLO:16}.  Underwater wireless optical communication (UWOC) is a promising technology for applications that allows a reliable communication link characterized by high channel capacity, low latency, energetic efficiency, and good communication range, which is up to $150$ {\si m} typically in clearwater \citep{HKBLO:16,FHM:15,HaR:08}. In \citep{WLGMR:17}, the authors developed a robust acquisition and tracking prototype in 3D underwater platforms for short-range communication. The technology used multiple photodetectors and a scanning strategy to exploit the variable beam divergence and provide robust acquisition for the underwater robots. However, mobile underwater platforms, including the effects of their dynamic model uncertainties and the uncertain underwater environment, are often limited by various operational and sensing capabilities constraints. These limitations induce a challenge for the trajectory tracking performance and increase the difficulty of maintaining a reliable optical communication.

Recent advances in low-cost light sources and more interest in ocean surveys have stimulated current optical communication research in an underwater area that has not been thoroughly investigated as terrestrial or free-space \citep{Cox:12}. Yet much work still needs to be achieved in system design and modeling and continued efforts to improve pointing error, localization, and tracking systems.
 
The underwater environment is challenging for all communication modes, with distinct tradeoffs between link range and data rate. The absorption and scattering effects of seawater, underwater turbulence, misalignment, and other factors can severely degrade the performance of underwater optical communication systems \citep{ZFZDC:17,SCAA:17}. These factors will result in frequent communication failures. Moreover, when transmitter and receiver mobility is involved, maintaining reliable optical-based communication by accurate pointing and tracking is getting even more challenging.

 \begin{figure}[!t]
\centering
      \begin{overpic}[scale=0.72]{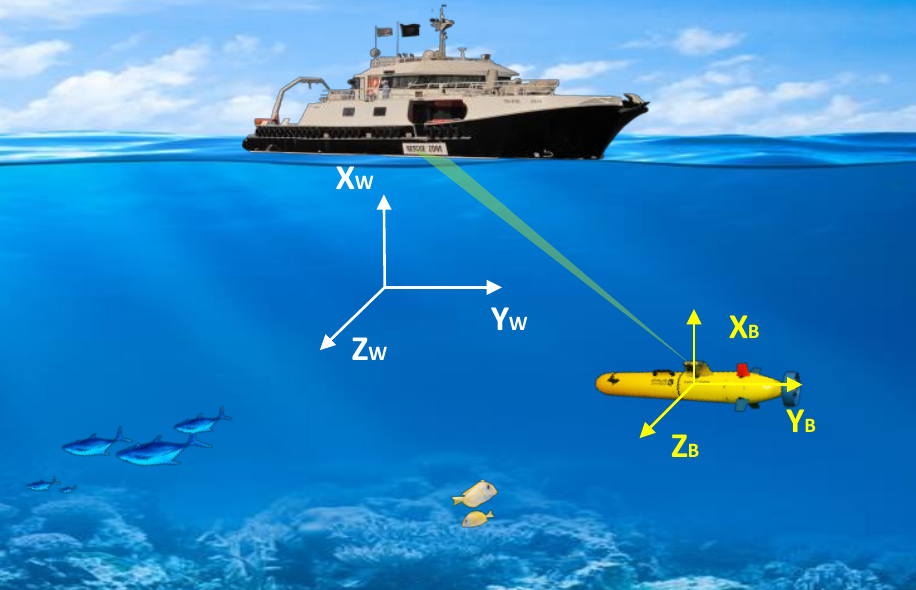}
       \put(78,12){\footnotesize  \textcolor{yellow}{$\{\cal B\}$}}
       \put(30,33){\footnotesize  \textcolor{white}{$\{\cal W\}$}}
            \end{overpic}  
              \caption{A team of surface-ship and autonomous underwater vehicle deployed to explore the ocean area.}\label{fig1}    
              \end{figure}

Underwater robotic vehicles have played a significant role in subsea marine operations, engineering, and science. The introduction of remotely operated vehicles (ROVs) has revolutionized the subsea industry. However, ROVs are tethered by an umbilical that makes their use complex, such as in restricted and disaster subsea areas. Furthermore, the autonomous underwater vehicles (AUVs) removed the need for a tether and their associated support, which endowed their manipulation capabilities. Modern AUVs can coordinate broad classes of underwater tasks and have many relevant applications in a versatile range from the deep sea to coastal water. AUVs integrate smart control algorithms and software schemes to select a proper communication strategy to apply in the communication system requirements and the actual environment status \citep{FSQSQ:15,SSB:17}. This development will enhance infrastructures for commercial use and scientific research by formulating versatile solutions to efficient communication between underwater vehicles, surface vessels, and seafloor infrastructures \citep{FBWPT:10,SSB:17,WKLS:19}.


Valid options for terrestrial communication on robotic platforms have been broadly investigated for underwater vehicles. However, most of the radio-frequency (RF) communications fail in the maritime area due to the significant absorption of electromagnetic radiation through saline water \citep{CrF:15},\citep{CHWF:18}.  Optical communication systems purpose a broad license-free spectrum, {\it, i.e.} the visible light and infrared spectrum, which contribute to a high modulation bandwidth and reliability \citep{CHWF:18}. From an economic perspective, optical components in optical communication systems are cheaper, lighter, and smaller than the high-speed RF systems. On the other hand, acoustic systems, which are dominant techniques for guidance and navigation in the underwater environment, are slow with high delay latency due to the slow propagation speed of sound under water.  Hence, optical wireless communications are of particular interest for data transmission in the underwater environment \citep{RuA:12,TZT:13,OLPKAO:15}. For example, \citep{DXR:13} performs an experimental test and demonstrates underwater wireless video communication in AUV based on a visible light communication system. Recently, a $7.2$ {\si Gbps} underwater optical wireless communication has been presented in \citep{WCWTL:17} for $450${\si nm} blue laser with a transmission distance of $6$ {\si m}. However, most existing underwater communication methods in the literature do not address that the AUV might require collecting and sending huge data such as images, sensor readings, or videos. Furthermore, these approaches deal with a particular concept of AUV trajectory tracking performance and focus on simple mapping of the underwater environment \citep{CZC:16,EZY:16,LiY:17,SSB:17c,SSB:17b,SSB:18,GTCC:20,ShS:20}. On the other hand, based on the existing trajectory tracking control approaches for underwater vehicles, there is a lack of an autonomous underwater control scheme that provides a complete and efficient control strategy in which the communication performance metrics are considered.

Differently from\citep{ZNBAAL:20}, the present work is not concerned with the localization problem. Indeed, a hybrid acoustic-optical communication is adopted in\citep{ZNBAAL:20}. Here we study the optical communication channel in which the position and linear velocity are constantly updated through the vehicle's onboard sensors. Additionally, the effect of the background noise is not considered in\citep{ZNBAAL:20}.\\  
This paper contains the following significant new contributions.
\begin{enumerate}
  \item The proposed NLPD controller aims at reinforcing the classical PD controller to achieve satisfactory robustness performance against mass parameter errors, measurement noises, and external disturbances forces.
  \item Mass parameter errors are added to the AUV system to test the proposed NLPD and PD controllers' robustness. 
  \item An analytical expression of the bit error rate that helps to compute the maximum achievable link distance is derived.
  \item Detailed proofs of the asymptotic stability for the trajectory tracking error problem using the proposed NLPD control are provided.
\end{enumerate}

The contributions of this paper are summarized as follows:
\begin{itemize}
  \item We study the trajectory tracking of a mobile ship receiver by an AUV transmitter to establish a directed optical line-of-sight (LoS) link, as illustrated in Fig. \ref{fig1}.
  \item We consider that both the surface ship and the AUV are equipped with only an optical communication system and describe the model of a direct optical LoS communication link between these systems in the presence of the solar background noise.
  \item We derive an analytical expression of the bit rate, compute the maximum achievable link distance, and propose a solution to keep an accurate position for UWOC systems while accomplishing a desired bit error rate under a solar background noise.
  \item  Finally, we propose a nonlinear proportional-derivative (NLPD) controller and a proportional-derivative (PD) controller to ensure this relative position is in the cone-shaped region. The simulation results show satisfactory robustness aspects of the NLPD controller.
\end{itemize}

This paper is organized as follows: Section \ref{sec-Hybrid-AO} details the model of the UWOC channel under solar radiation underwater. Also, a cone-shaped beam profile of the mobile surface ship receiver is defined based on the optical channel model. Hence, if the AUV transmitter lies in this cone-shaped region, the directed LoS optical link is maintained with the desired minimum rate. Section \ref{method} provides the methodology to show that the AUV tracks the mobile ship receiver and lies in this cone-shaped region. Section \ref{simulations} outlines the results of the numerical simulations obtained using the MATLAB/Simulink. Finally, section \ref{conclusion} concludes the paper.


\section{Optical link modeling}\label{sec-Hybrid-AO}
Let us consider a channel model for the optical communication link to predict the effect of solar background noise on the optical link budget and then derive the maximum link distance to achieve a specified bit error rate (BER) and bit rate.

\subsection{Underwater Optical Communication Link}\label{sec-optical-link}
We use the intensity modulation with direct detection (IM/DD) scheme that can produce high-speed links for a range of systems via on-off Keying (OOK) in the physical layer \citep{GPR:12}. On the other hand, we suppose a directed LoS optical link model of IM/DD system \citep{GPR:12}.  Fig.~\ref{fig2} illustrates the
main parameters of this optical channel model. 

The radiant intensity is given by~\citep{GPR:12},\citep{SDDS:14}
\begin{equation}\label{eq-OW-1}
I_{s}(d,\phi)=P_{{\scriptsize\mbox{TX}}}\frac{m+1}{2\pi d^2}\cos^m \phi,
\end{equation}
where $d$ is the distance between transceivers installed on the ship and AUV, $\phi$ represents the pointing angle with respect to the optical link,  $P_{{\scriptsize\mbox{TX}}}$ is the average power emitted by transmitter, and $m$ is the Lambert's mode number representing the effect of the source beam. This number is given as \citep{GPR:12},\citep{SDDS:14}
\begin{equation}\label{eq-OW-2} 
m=\frac{-\ln 2}{\ln (\cos \Phi_{1/2})},
\end{equation}
where $\Phi_{1/2}$ describes the half-angle at half-power of average transmitted optical source which models the transmitter beam width.

\begin{figure}[!t]
\centering
      \begin{overpic}[scale=0.72]{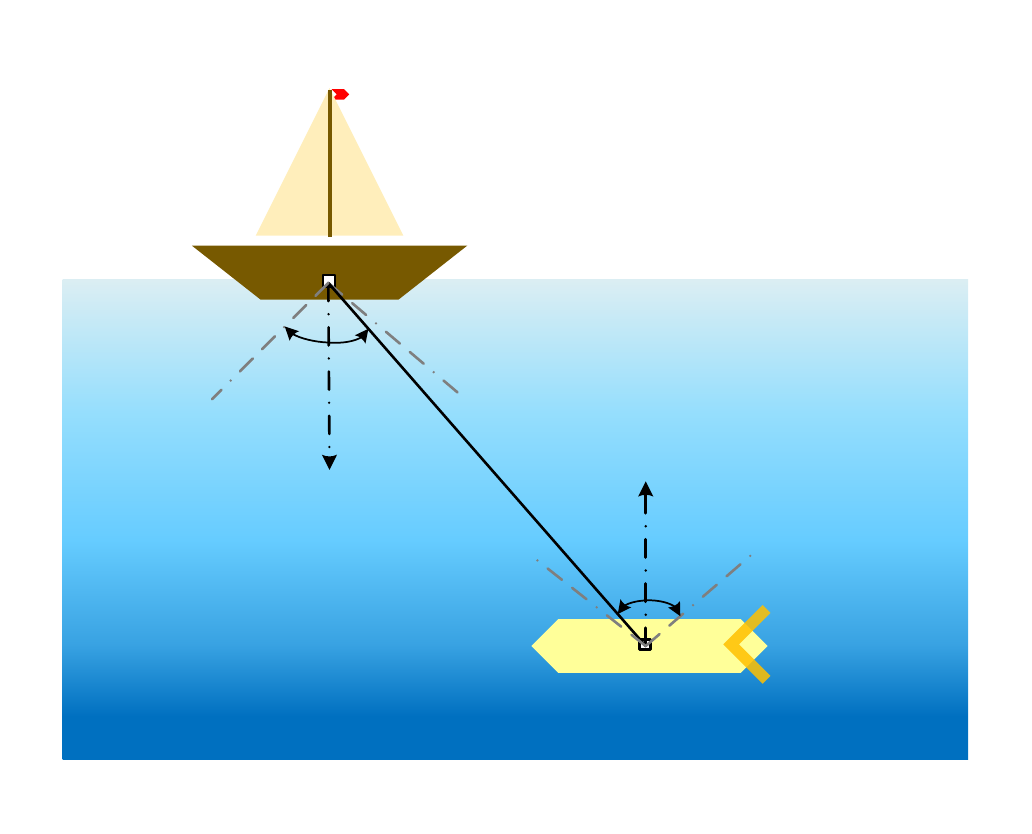}
       \put(60,21){\footnotesize  $\phi$}
         \put(65,22){\footnotesize  $\Phi_{\frac{1}{2}}$ }
          \put(46,36){\footnotesize  $d$}
 \put(32,43){\footnotesize  $\psi$}
  \put(24,44){\footnotesize $\Psi_{\cal C}$ }
            \end{overpic} 
              \caption{The parameters of the optical channel model.}\label{fig2}    
\end{figure}
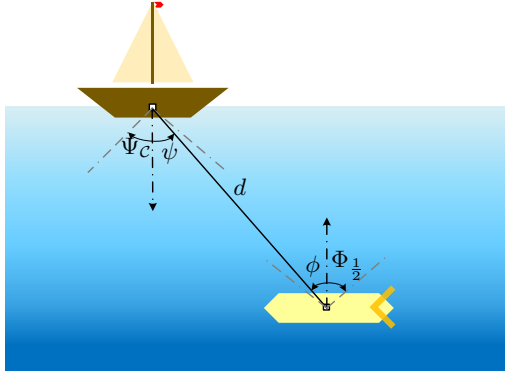

The optical receiver can be described as an effective area $A_{\mbox{\scriptsize eff}}$ which collects the incident light at the angle $\psi$. This area is defined as 
\begin{equation}\label{eq-OW-3}
A_{\mbox{\scriptsize eff}}(\psi)=f(\psi)A_r \cos \psi, \quad \abs{\psi}\leqslant\Psi_{\cal C},
\end{equation}
where $A_r$ is the detector active area, and $f(\psi)$ mimics the light concentrator gain for a standard case which is given by
\begin{equation}\label{eq-OW-4}
f(\psi)=
\left\{\begin{array}{llll}
\displaystyle \frac{n^2}{\sin^2 \Psi_{\cal C}}\qquad\quad \mbox{if}\, \abs{\psi}\leqslant\Psi_{\cal C},\\
0 \qquad\qquad\qquad  \mbox{if}\, \abs{\psi}>\Psi_{\cal C}.
\end{array}\right.
\end{equation}
Here, $\Psi_{\cal C}$ is the half-angle field-of-view (FOV) of the optical receiver, and $n$ is the refractive index of the sea water.

Marine light propagation conforms to high attenuation caused by scattering and absorption. Those phenomena depend on the light beam and water properties. We use an exponential attenuation model to describe it, which provides an accurate estimate of the optical power in clear ocean waters where the absorption is predominant\citep{Mob:94},\citep{GKBLR:13}.

We settle the formulation of the channel loss as follows
\begin{equation}\label{eq-OW-5a}
L_{\mbox{\scriptsize ch}}=\exp(-K_{a} \bar{d}),
\end{equation}
where $K_{a}$ is the attenuation coefficient\citep{Mob:94}, and $\bar{d}$ denotes the distance between the transmitter and receiver. We assume that $\phi=\psi$, {\it i.e.} the optical transmitter attached on the AUV is always pointing up, while the one in the surface ship is always pointing down, as illustrated in Fig. \ref{fig2}.

Using \eqref{eq-OW-1}, \eqref{eq-OW-3} and \eqref{eq-OW-5a}, the optical signal strength at the receiver is given by\citep{GPR:12}
\begin{equation}\label{eq-OW-5}
P_{\mbox{\scriptsize RX}}=I_s A_{\mbox{\scriptsize eff}}L_{\mbox{\scriptsize ch}},
\end{equation}
and the average corresponding photocurrent is computed by
\begin{equation}\label{eq-OW-6}
I_{\mbox{\scriptsize b}}=RP_{\mbox{\scriptsize RX}},
\end{equation}
where $R$ is the responsivity of the receiver. 

\subsection{Effect of the solar radiation noise on the underwater optical link performance}\label{sec-optical-power-indicent}
The impact of solar radiation's background noise on the UWOC link operating in clear ocean waters can limit the BER performance for relatively low depths\citep{HKBLO:16}. Most previous works in the literature neglect the effect of solar noise. However, the accuracy of negligible background noise has not been investigated for low depths. Previous results reveal that the light of blue-green wavelengths can indeed pass through the water due to high water transparency in clear ocean waters\citep{HKBLO:16}. To investigate solar radiation's effect on the underwater optical link performance, we consider the case of clear water when sunlight passes through water and reaches a considerable level. The solar irradiance attenuates exponentially and becomes progressively diffuse with the depth. Thus, the solar spectral downwelling plane irradiance is given as\citep{HKBLO:16},\citep{MBR:16}
\begin{equation}\label{eq-OW-a1}
E_{\mbox{\scriptsize s}}(\lambda,\bar{d})=E_{\mbox{\scriptsize s}}(\lambda,0)\exp(-K_a\bar{d}),
\end{equation}
where $\bar{d}$ and $\lambda$ denote the operation depth and the wavelength, respectively. We suppose that the effect of obstruction at the detector side to restraint the solar irradiance is neglected and the receiver has a bandpass filter with bandwidth $\Delta \lambda$. Then, the optical power of the solar noise through the receiver is equal to
\begin{equation}\label{eq-OW-a2}
P_{\mbox{\scriptsize b}}=E_{\mbox{\scriptsize s}}(\lambda_0,\bar{d})\epsilon_t\Delta\lambda\exp(-K_a\bar{d})A_{\mbox{\scriptsize eff}}.
\end{equation}
Here, $\epsilon_t$ represents the water transmittance and is typically set to $95\%$ for $\lambda_0=532${\si nm} based on the experimental data given in\citep{HKBLO:16},\citep{MBR:16}, and $E_{\mbox{\scriptsize s}}(\lambda_0,0)=0.7645${\si W/m$^2\times$nm} is the solar noise power in the sea surface. This solar noise is usually described by a Gaussian process\citep{MBR:16} whose variance is estimated by
\begin{equation}\label{eq-OW-a3}
\sigma_{\mbox{\scriptsize b}}^2=2eRP_{\mbox{\scriptsize b}}B,
\end{equation}
where $B$ and $e$ denote the bit rate and the charge of an electron, respectively.

The optical channel is supposed to adopt IM/DD based on OOK, then the signal-to-noise ratio (SNR) at the detector side is given by~\citep{GPR:12},\citep{KaB:97}
\begin{equation}\label{eq-OW-a4}
\mbox{SNR}=\frac{I_{\mbox{\scriptsize b}}^2}{\sigma_{\mbox{\scriptsize b}}^2}.
\end{equation}

The SNR for OOK is proportional to the BER by\citep{GPR:12,CrF:15,YCT:14,CHWF:18}
\begin{equation}\label{eq-OW-a5}
\mbox{BER}=Q\Big(\sqrt{\mbox{SNR}}\Big),
\end{equation}
where $Q(.)$ is the tail probability of the normal distribution Q-function which is given by
\begin{equation*}
\displaystyle Q(x)\!\!=\!\!\!\int_x^{+\infty}\!\!\!\exp(-t^2/2)\der\! t.
\end{equation*}
For a given BER and combining \eqref{eq-OW-a3} and \eqref{eq-OW-a4}, the bit rate $B$ can be calculated by
\begin{equation}\label{eq-OW-a6}
B=\frac{1}{2eRP_{\mbox{\scriptsize b}}}\left[\frac{I_{\mbox{\scriptsize b}}}{Q^{-1}(BER)}\right]^2,
\end{equation}
where $I_{\mbox{\scriptsize b}}$ can be computed using \eqref{eq-OW-1} -- \eqref{eq-OW-6} and $P_{\mbox{\scriptsize b}}$ is given by \eqref{eq-OW-a2}. 

\subsection{Cone-Shaped Beam Region}\label{sec-cone}
Figs. \ref{fig-contour}\textbf{a)} and \ref{fig-contour}\textbf{b)} show the contour plots of the optical link bit rate in logarithm scale with respect to the transmitter-receiver position by assuming $\ber = 10^{-4}$. Table \ref{sample-table-OW} gives the optical parameters used to plot these contour maps. As shown in Fig. \ref{fig-plot1}\textbf{a)}, we define a right circular region called cone-shaped region using this optical range such that the maximum receiver pointing error yields when $\psi=\Psi_{\cal C}$. Consequently, a received signal strength greater or equal to the threshold one will be sufficiently achieved if the AUV transmitter stays inside this region. The half vertex of $\cal C$ is given by the position of the ship on the water surface. We define the vector of the normal direction of the cone axis as follows $\textbf{e}=\begin{bmatrix} 0 & 0 & 1 \end{bmatrix}^T$\!\!. Hence, $\cal C$ is defined by its half aperture angle $\Psi_{\cal C}$ and by its slant height $d_{\cal C}$. Since $\ib$ is a transcendental function of $d$, the value of the slant height ${d_{\cal C}}$ of ${\cal C}$ is evaluated numerically through the intersection of the $\bitrate$ distribution along the distance $d$ as shown in Fig. \ref{fig-plot1}\textbf{b)}. Finally, we
can find that this bit rate $\bitrate\!=\!10${\si Mbps} which is equivalent at $``7"$ in logarithm scale is maintaining at a range of $\maxd \approx 4.4${\si m}. Within this range, we obtain that the height $h$ of the cone which is $h_{\cal C}\!=\!d\!\times\!\cos \Psi_{\cal C}\!=\!3.8157${\si m}.

\begin{table}[!t]
\caption{Parameters of the optical link.}
\begin{center}
\begin{tabular}{| c || c c c c|}
  \hline
Transmitter  & $P=0.1${\si W}~~  & ~$K_a=0.15${\si m$^{-1}$}~   & $\Phi_{1/2}=15^{\circ}$  & $\Delta\lambda=30${\si nm}     \\ \hline
~Receiver~~  &  $A_r=1${\si cm}$^2$ & $\Psi_{\cal C}=30^{\circ}$ & $~n=1.52~$ & ~$R=0.6${\si A/W}     \\ \hline
\end{tabular}
\end{center}
\label{sample-table-OW}
\end{table}

\begin{figure*}[!t]
   \begin{minipage}[c]{0.45\linewidth}  
           \centering
      \begin{overpic}[scale=0.26]{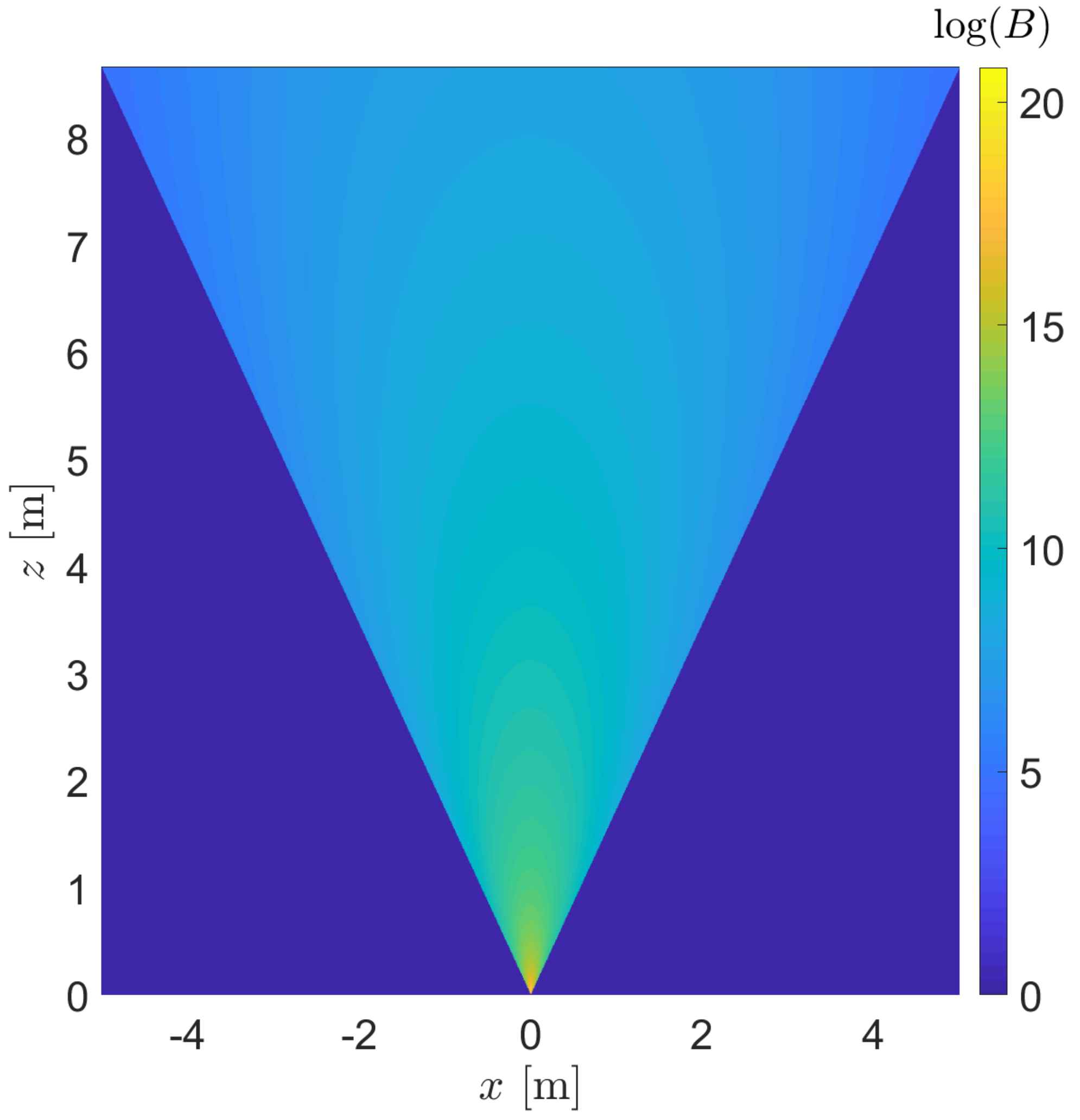}
      \put(20,0){\scriptsize  \textbf{a)}}
            \end{overpic}  
       \end{minipage}\hfill 
         \begin{minipage}[c]{0.48\linewidth}
         \centering
      \begin{overpic}[scale=0.32]{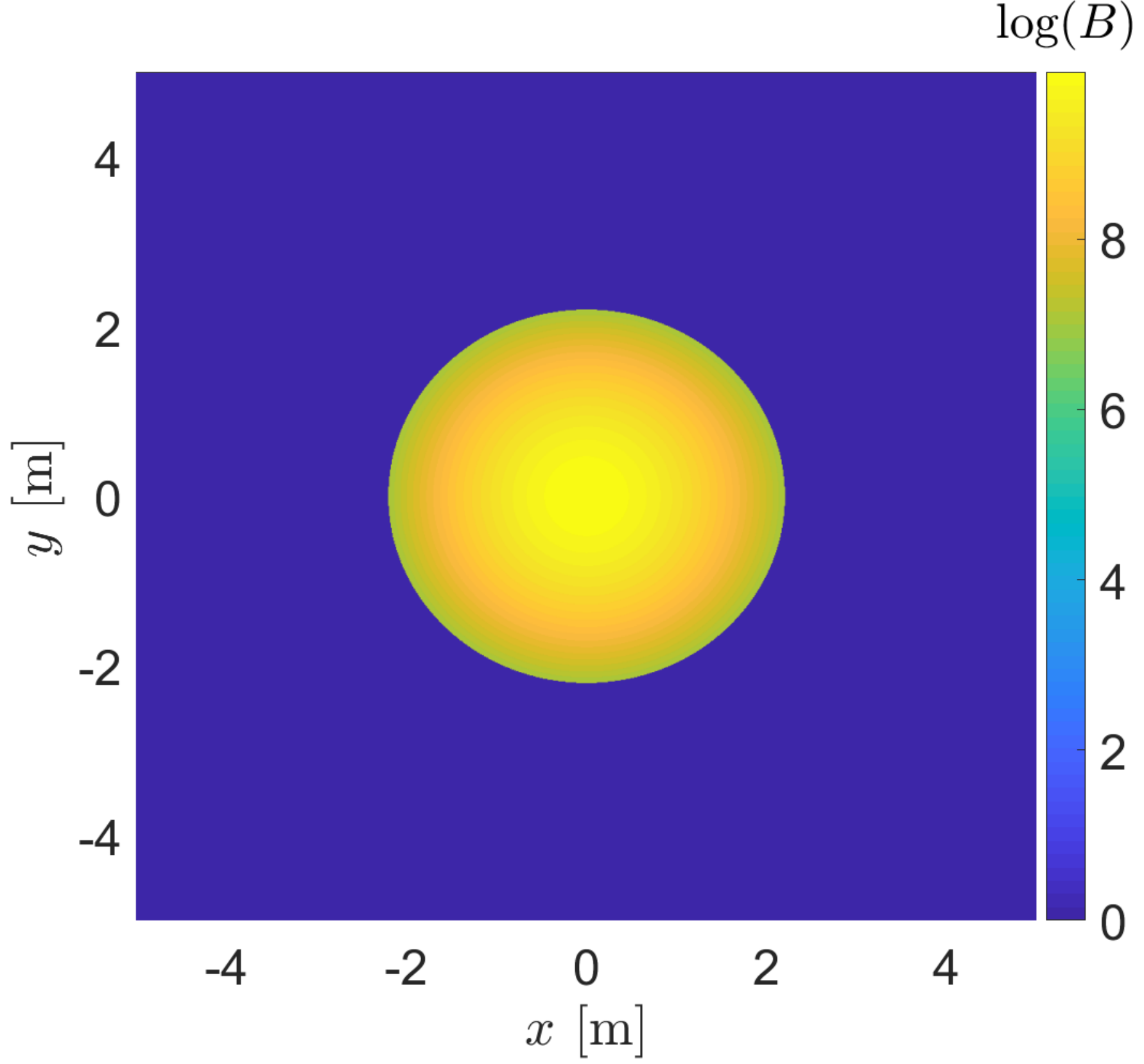}
      \put(20,-1){ \scriptsize \textbf{b)}}
            \end{overpic}              
      \end{minipage}            
      \caption{Contour maps of bit rates in logarithm scale ($\bitrate\!=\!10$ {\si Mbps} of bit rate is equivalent at $``7"$ in logarithm scale): \textbf{a)} Side view;  \textbf{b)} Top view.} \label{fig-contour}
\end{figure*}

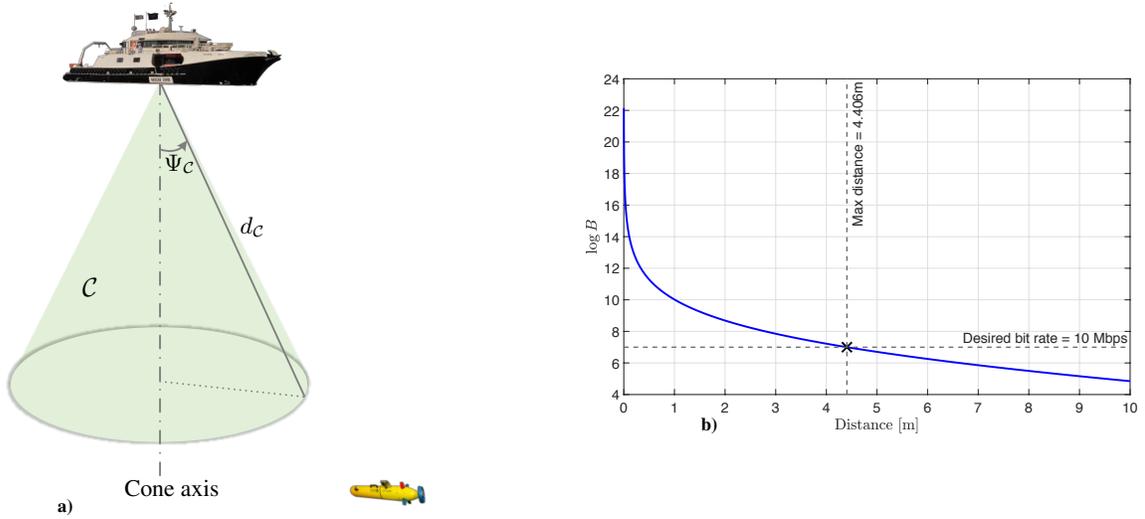
\begin{figure*}[!t]
   \begin{minipage}[c]{0.45\linewidth} 
   \centering 
      \begin{overpic}[scale=0.70]{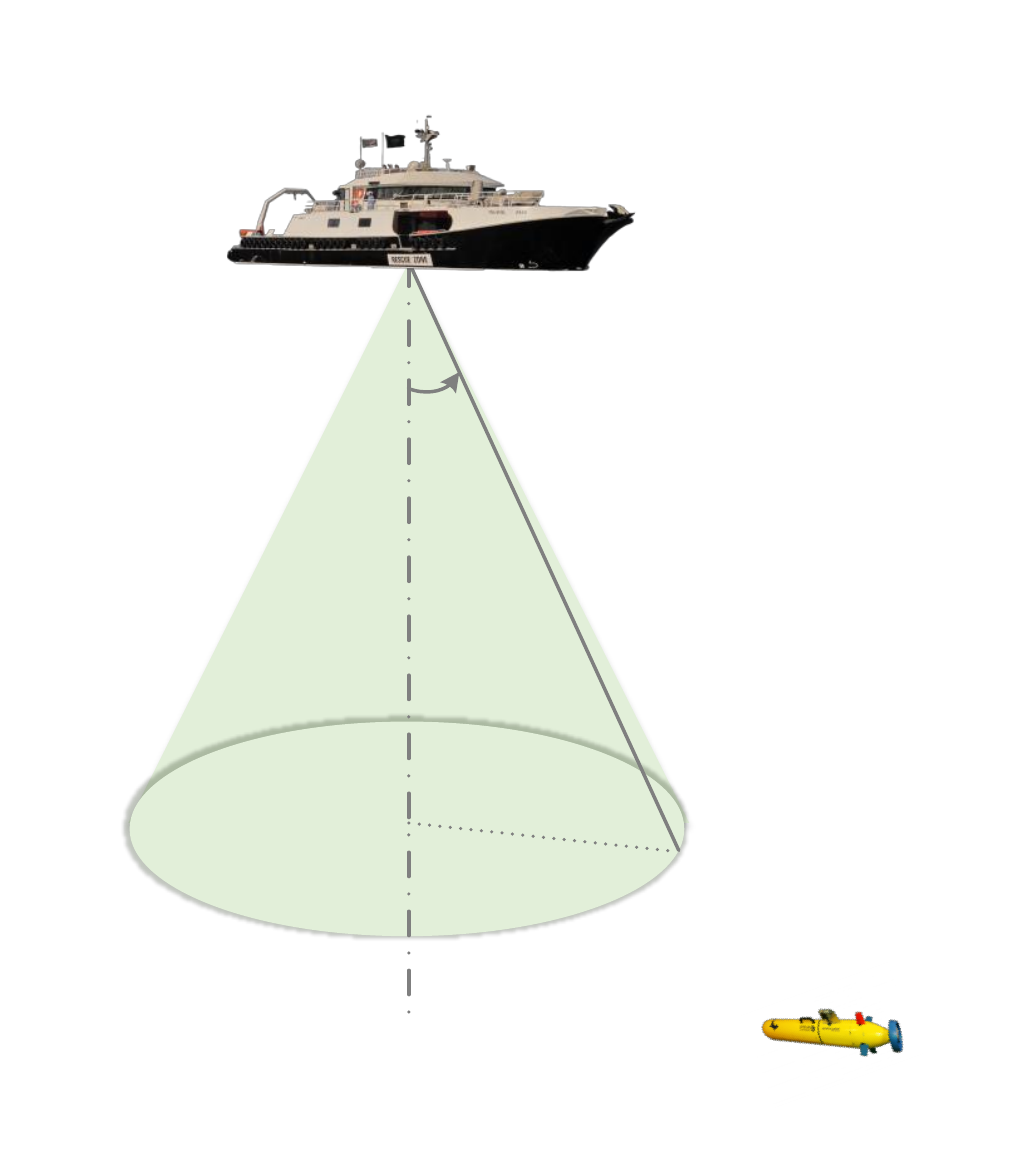}
            \put(16,42){$\cal C$}
  \put(31.5,66){\footnotesize $\Psi_{\cal C}$ }
    \put(46,54){\footnotesize $d_{\cal C}$ }
        \put(24,4){\footnotesize Cone axis }
        \put(10,1){ \scriptsize \textbf{a)}}
            \end{overpic}  
       \end{minipage}\hfill
         \begin{minipage}[c]{0.48\linewidth}
      \begin{overpic}[scale=0.24]{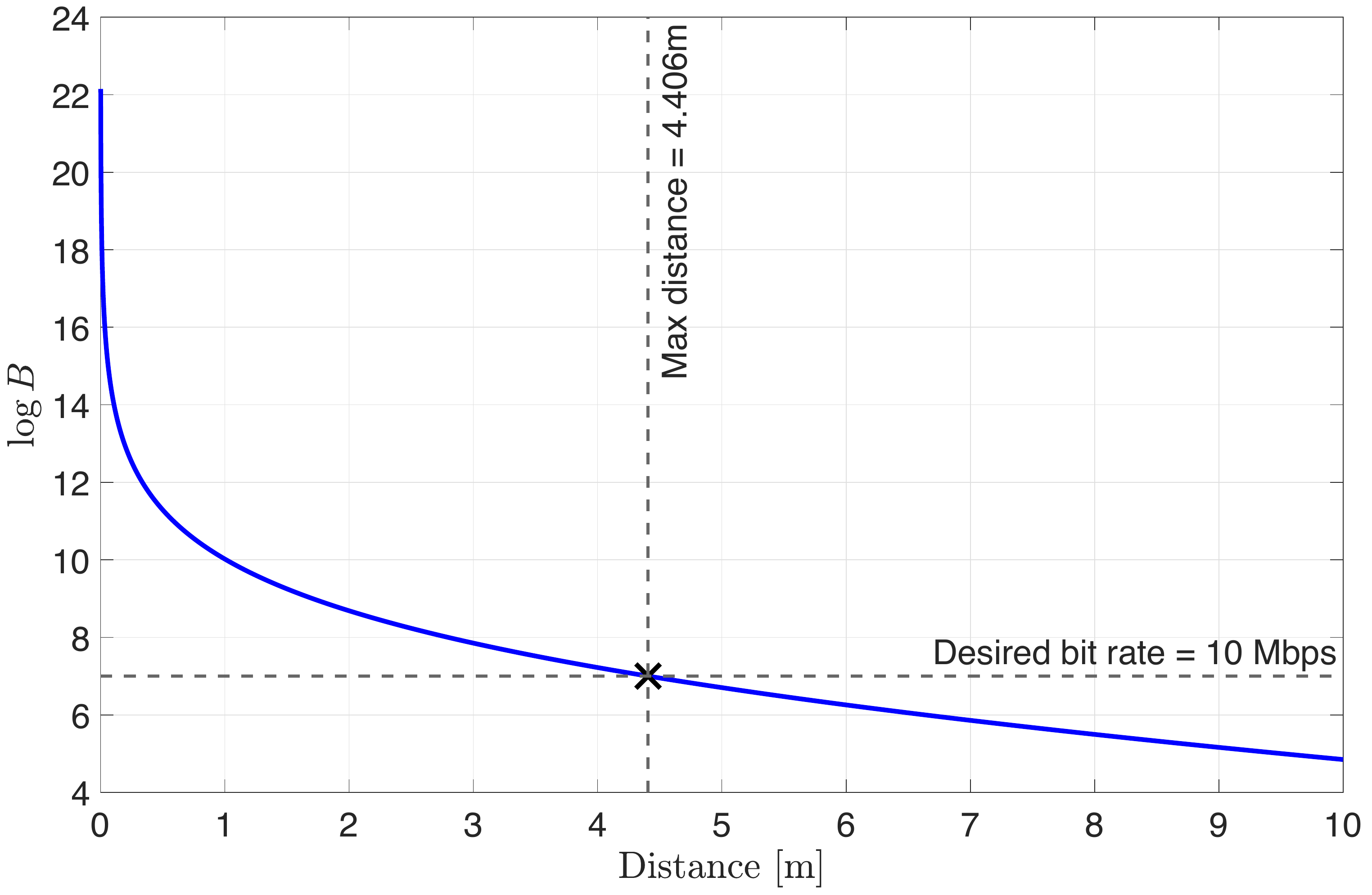}
      \put(20,1){ \scriptsize \textbf{b)}}
            \end{overpic}  
   \end{minipage}    
      \caption{\textbf{a)} Cone-shaped beam region $\cal C$ with its main parameters;  \textbf{b)} Slant height value ${d_{\cal C}}$ of ${\cal C}$.} \label{fig-plot1}
\end{figure*}

\section{Methodology}\label{method}
Here, our strategy is to control the AUV to drive it inside the cone-shaped beam region $\cal C$ and then maintain it there. First, the dynamic model of the AUV is introduced. Then, we derive two different controllers for the AUV transmitter such that it holds a good position with regards to the mobile ship receiver and stays within the cone-shaped beam region $\cal C$ to guarantee an optical wireless communication link with the desired bit rate. 

\subsection{Model of the Autonomous Underwater Vehicle (AUV)}
We define an earth-fixed $\{\cal W\}$ reference and a body-fixed reference $\{\cal B\}$ for the motion control of the AUV, as shown in Fig. \ref{fig-UnderW}. The AUV dynamics in $\{\cal B\}$ frame is given as \citep{Fos:94}, \citep{CGEC:10}
\begin{equation}\label{eq-Mod1}
M\dot{\nu}+C(\nu)\nu+D(\nu)\nu+g(\eta)=\tau+\tau_w,
\end{equation}
where $\eta$ and $\nu$ are the position and velocity of the AUV in $\{\cal W\}$ and $\{\cal B\}$, respectively. $M$ represents the inertia matrix, $C(\nu)$ is  the matrix of Coriolis' and centripetal terms and $D(\nu)$ represents the damping matrix. $\tau$ is the input vector signals, $g(\eta)$ is an unknown vector of restoring forces and $\tau_w$ is the bounded external input (force/torque) disturbances. Table \ref{para}  provides the main parameters of the AUV model \citep{CGEC:10}.

\begin{figure}[!t]
\centering
      \begin{overpic}[scale=0.26]{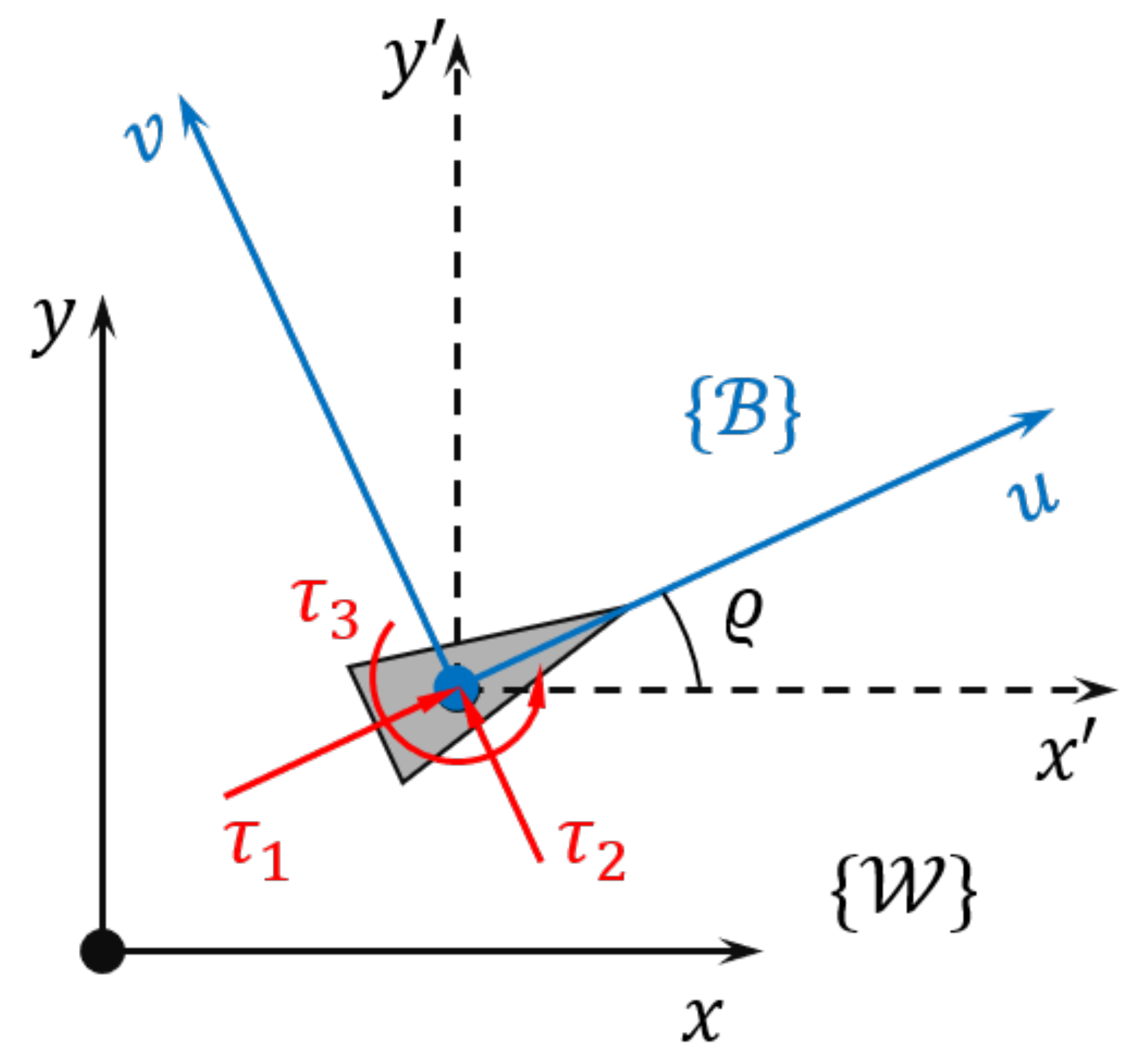}
            \end{overpic}  
              \caption{Underwater vehicle model in horizontal plane.}\label{fig-UnderW}    
\end{figure}

\begin{table}[!t]
\caption{Nomenclature.}
\begin{center}
\begin{tabular}{| c || c || c |}
  \hline 
\textbf{~~Symbol~~}  & \textbf{~~~Description~~~}  & \textbf{~~Unit~~}   \\ \hline  
$x$ &  Position in surge & \si{m} \\  
~~$y$~~ &Position in sway& \si{m} \\  
~~~$\varrho$~~~ &Yaw &\si{rad}\\  \hline
$u$ &  ~~~~Linear velocity in surge~~~~  &~~~~\si{m}\si{s^{-1}}~~~~  \\  
$v$ & ~~Linear velocity in sway~~ &\si{m}\si{s^{-1}} \\  
$r$ &~~Yaw velocity  &\si{rad}\si{s^{-1}} \\  \hline
$\tau_1$ &  Force in surge & \si{N}   \\  
$\tau_2$ &~Moments in yaw~ & \si{N}\si{m} \\  
$\tau_{{w}_1}$&~Disturbance in surge~ &\si{N}\\  
$\tau_{{w}_2}$ &~Disturbance in sway~ &\si{N}\\  
$\tau_{{w}_3}$ &~Disturbance in yaw~ &\si{N}\si{m}\\  \hline
$m$ &  Weight of the AUV & \si{kg}  \\  
$I_z$ &Moments of inertia in yaw &\si{kg}\si{m^2}\\   \hline
$X_{\dot{u}}$ &Added mass in surge &\si{kg}\\  
$Y_{\dot{v}}$ &  Added mass in sway  & \si{kg} \\  
$N_{\dot{r}}$ &  Added mass in yaw & \si{kg}\si{m^2} \\   \hline
$X_{u}$ &Linear damping coefficients in surge  &\si{kg}\si{s^{-1}}\\
$Y_{v}$ &Linear damping coefficients in sway  &\si{kg}\si{s^{-1}}\\
$N_{r}$ &  Linear damping coefficients in yaw  &\si{kg}\si{m^2}\si{s^{-1}}\\\hline
$X_{u\abs{u}}$ &~~~Quadratic damping coefficients in surge~~~  &\si{kg}\si{m^{-1}}\\
$Y_{v\abs{v}}$ &Quadratic damping coefficients in sway  &\si{kg}\si{m^{-1}}\\
$N_{r\abs{r}}$ & Quadratic damping coefficients in yaw &\si{kg}\si{m^{2}}\\ \hline
$\cal W$ & Earth-fixed frame &\\ 
$\cal B$ & Body-fixed frame &\\ \hline
\end{tabular}
\end{center}
\label{para}
\end{table}

We consider the AUV transmitter in the horizontal plane at a constant depth, such as ocean floor applications. Then, 
\begin{align*}
& \eta=\begin{bmatrix} x & y & \varrho \end{bmatrix}^T\!\!, \qquad \nu=\begin{bmatrix} u& v & r \end{bmatrix}^T\!\!,  \qquad \tau\!=\!\begin{bmatrix} \tau_1 & \tau_2  & \tau_3 \end{bmatrix}^T\!\!, 
\notag \\& w\!=\!\begin{bmatrix} w_1 & w_2 & w_3 \end{bmatrix}^T\!\!, \qquad g(\eta)\!=\!0, \qquad  M\!=\!\begin{bmatrix} m_{11}& 0 & 0 \\ 0 & m_{22} & 0 \\ 0 & 0 & m_{33}\end{bmatrix}\!\!,  \notag \\& D(\nu)\!=\!\!\!\begin{bmatrix} d_{11}& 0 & 0 \\ 0 & d_{22} & 0 \\ 0 & 0 & d_{33}\end{bmatrix}\!\!, \qquad C(\nu)\!=\!\!\!\begin{bmatrix} 0& 0 & -m_{22}v \\ 0 & 0 & m_{11}u \\ m_{22}v & -m_{11}u & 0\end{bmatrix}\!\!,
\end{align*}
with $m_{11}=m-X_{\dot{u}}$, $m_{22}=m-Y_{\dot{v}}$, $m_{33}=I_z-N_{\dot{r}}$, $d_{11}=-X_{u}-X_{u\abs{u}}\abs{u}$, $d_{22}=-Y_{v}-Y_{v\abs{v}}\abs{v}$, and $d_{33}=-N_{r}-N_{r\abs{r}}\abs{r}$ \citep{CGEC:10,AgP:07,LaS:07}. 

The AUV position in initial coordinate ${\cal W}$ frame is given as follows
\begin{equation}\label{eq-Mod2}
\dot{\eta}=R(\varrho)\nu,
\end{equation}
with the transformation matrix $R$ is given as
\begin{equation}\label{eq-Mod2a}
R(\varrho)\!=\!\begin{bmatrix} \cos\varrho & -\sin\varrho & 0\\ \sin\varrho & \cos\varrho & 0\\ 0 & 0& 1\end{bmatrix}.
\end{equation}

From the kinematic transformations of the state variables and the model parameter, the formulation of the AUV dynamics is transformed to ${\cal W}$ frame as follows 
\citep{Fos:02},\citep{CCCTL:17}
\begin{align}\label{eq-Mod2aa}
&\ddot{\eta}=R(\varrho)\dot{\nu}+\dot{R}(\varrho)\nu,\notag \\
&M(\eta)=R^{-T}(\varrho)MR^{-1}(\varrho), \notag \\
&C(\nu,\eta)=R^{-T}(\varrho)\left[C(\nu)-MR^{-1}(\varrho)\dot{R}(\varrho)\right]R^{-1}(\varrho),\notag \\
&D(\nu,\eta)=R^{-T}(\varrho)D(\nu)R^{-1}(\varrho), \notag\\
&g(\eta)=R^{-T}(\varrho)g(\eta),\notag\\
&\tau(\eta)=R^{-T}(\varrho)\tau.
\end{align}

\subsection{Tracking Control Strategy}
The control of AUVs has several issues due to the nonlinear dynamics and model uncertainties\citep{Ver:09}. The proportional derivative (PD) controller is the most used technique to control the dynamic marine vehicles due to its design simplicity and performance results. However, it lacks better performance results in the presence of nonlinearity or disturbances in the system dynamics\citep{Ver:09}. To overcome this drawback, we propose a NLPD controller to reinforce the classical PD controller by implementing a saturation function to achieve the tracking problem's robustness performance. 

Our objective is to design a nonlinear proportional derivative input $\bf{\tau}$ based on proportional-derivative control for trajectory tracking such that the AUV enters first in the cone-shaped region ${\cal C}$ and then stays in the area.

We  consider the state of the mobile ship receiver as following $\begin{bmatrix}x_{\mbox{\tiny RX}}& y_{\mbox{\tiny RX}}& \varrho_{\mbox{\tiny RX}} \end{bmatrix}^T$, where $\eta_{\mbox{\tiny RX}}\!=\!\begin{bmatrix}x_{\mbox{\tiny RX}}& y_{\mbox{\tiny RX}}\end{bmatrix}^T$ is the position of the surface ship receiver in $x\!-\!y$ plane and $\varrho_{\mbox{\tiny RX}}$ is its heading angle. We assume that the surface ship vehicle can change its heading angle $\varrho_{\mbox{\tiny RX}}$, and moves around on the sea surface area with a constant position and linear velocity given by $\eta_{\mbox{\scriptsize ref}}$ and $\bf{v_{\mbox{\scriptsize ref}}}$, respectively. The position of the AUV is given by $\eta=\begin{bmatrix}x& y& \varrho \end{bmatrix}^T\in \R^3$. The NLPD controller requires full state measurement. Here, we consider that the AUV knows continuously the position of the mobile ship and its velocity. We propose the NLPD controller for the trajectory tracking as follows\citep{GTCC:19}
\begin{align}\label{eq-OW-aa}
 {\bf \tau}= R^T(\varrho) &\Big[ M(\eta)\ddot{\eta}_{\mbox{\scriptsize ref}}+C(\nu,\eta)\dot{\eta}_{\mbox{\scriptsize ref}}+D(\nu,\eta)\dot{\eta}_{\mbox{\scriptsize ref}} +K_p(.)\widetilde{\eta}+K_v(.)\widetilde{\nu}\Big],
\end{align}
where $\widetilde{\eta}=\eta_{\mbox{\scriptsize ref}}-\eta$ and $\widetilde{\nu}=\bf{v}_{\mbox{\scriptsize ref}}-\nu$ are the position and velocity errors, respectively. $K_p$ and $K_v$ are the control gains, $\bf{v}_{\mbox{\scriptsize ref}}$ is the reference velocity, and $\eta_{\mbox{\scriptsize ref}}$ is the reference position which has to be determined.

The following theorem provides the conditions that guarantee the $\cHi$-stability for the trajectory tracking system.
\begin{theorem}\label{thm}
The NLPD controller \eqref{eq-OW-aa} asymptotically stabilizes system \eqref{eq-Mod1}-\eqref{eq-Mod2} for trajectory tracking where the controller gain matrices $K_p(.)$ and $K_v(.)$ have the following form
\begin{equation*}
K_p(.)=\diag\begin{bmatrix} k_{p1}(.)  & k_{p2}(.)  & k_{p3}(.)  \\  \end{bmatrix}>0,
\end{equation*}
\begin{equation*}
K_v(.)=\diag\begin{bmatrix} k_{v1}(.) & k_{v2}(.)  &  k_{v3}(.)  \\  \end{bmatrix}>0,
\end{equation*}
with $k_{pj}(.)$ and $k_{vj}(.)$ are defined as follows
\begin{equation*}
k_{pj}(.)= 
\begin{cases}
    a_{pj}\abs{\widetilde{\eta}_j(t)}^{(\mu_{pj}-1)},& \text{if } \abs{\widetilde{\eta}_j(t)}\geqslant b_{pj}\\
    a_{pj}b_{pj}^{(\mu_{pj}-1)},              &  \text{if } \abs{\widetilde{\eta}_j(t)}\geqslant b_{pj}\\
     & \forall \mu_{pj} \in [0 \,, 1]
\end{cases}
\end{equation*}
\begin{equation*}
k_{vj}(.)= 
\begin{cases}
    a_{vj}\abs{\widetilde{\nu}_j(t)}^{(\mu_{vj}-1)},& \text{if } \abs{\widetilde{\nu}_j(t)}\geqslant b_{vj}\\
    a_{vj}b_{vj}^{(\mu_{vj}-1)},              &  \text{if } \abs{\widetilde{\nu}_j(t)}\geqslant b_{vj}\\
    & \forall \mu_{vj} \in [0 \,, 1]
\end{cases}
\end{equation*}
and $a_{pj}$, $a_{vj}$, $b_{pj}$ and $b_{vj}$ are positive constants.
\end{theorem}
\begin{proof}
By substituting the NLPD control law \eqref{eq-OW-aa} in \eqref{eq-Mod1}-\eqref{eq-Mod2} leads to 
\begin{equation*}
M(\eta)\ddot{\widetilde{\eta}}=-C(\nu,\eta)\dot{\widetilde{\eta}}-D(\nu,\eta)\dot{\widetilde{\eta}}-K_p(.)\widetilde{\eta}-K_v(.)\dot{\widetilde{\eta}}+\tau_{w},
\end{equation*}
which can be rewritten as
{\small
\begin{equation}\label{eq-1}
\begin{bmatrix}
\dot{\widetilde{\eta}}\\ \ddot{\widetilde{\eta}}
\end{bmatrix}\!\!=\!\!
\begin{bmatrix}
\dot{\widetilde{\eta}}\\
\!-M(\eta)^{-1}\!\Big(C(\nu,\eta)\dot{\widetilde{\eta}}\!\!+\!D(\nu,\eta)\dot{\widetilde{\eta}}\!+\!K_p(.)\widetilde{\eta}\!+\!K_v(.)\dot{\widetilde{\eta}}\!-\!\tau_{w}\Big)
\end{bmatrix}\!\!.
\end{equation}
}

Consider the following Lyapunov candidate function\citep{KeC:96}
\begin{equation}\label{lyap1}
V(\widetilde{\eta},\dot{\widetilde{\eta}})=\frac{1}{2}\dot{\widetilde{\eta}}^TM(\eta)\dot{\widetilde{\eta}}+\int_{0}^{\widetilde{\eta}} \zeta^TK_p(\zeta)\der \zeta,
\end{equation}
where
\begin{align*}
\int_{0}^{\widetilde{\eta}}\! \!\! \zeta^TK_p(\zeta)\der \zeta\! \!=\!\!\int_{0}^{\widetilde{\eta}_1}\!\! \!\zeta_1^Tk_{p1}(\zeta_1)\der \zeta_1+\int_{0}^{\widetilde{\eta}_2} \! \!\! \zeta_2^Tk_{p2}(\zeta_2)\der \zeta_2+\int_{0}^{\widetilde{\eta}_3}\!\! \!\zeta_3^Tk_{p3}(\zeta_3)\der \zeta_3.
\end{align*}

The function $V(\widetilde{\eta},\dot{\widetilde{\eta}})$ is radially unbounded and positive definite\citep{KeC:96,CCCTL:17,GTCC:19}. The first term $\dot{\widetilde{\eta}}^TM(\eta)\dot{\widetilde{\eta}}$ is a positive definite function with respect to $\dot{\widetilde{\eta}}$ since $M(\eta)$ is a positive definite matrix. The integral term $\int_{0}^{\widetilde{\eta}} \zeta^TK_p(\zeta)\der \zeta$ is a potential energy induced by the position error-driven part of the controller. It is radically unbounded positive definite function\citep{KeC:96}. By calculating the time derivative of the Lyapunov candidate function \eqref{lyap1} along the trajectories of \eqref{eq-1}, we obtain
\begin{equation}\label{lyap2}
\dot{V}(\widetilde{\eta},\dot{\widetilde{\eta}})=\dot{\widetilde{\eta}}^TM(\eta)\ddot{\widetilde{\eta}}+\frac{1}{2}\dot{\widetilde{\eta}}^T\dot{M}(\eta)\dot{\widetilde{\eta}}+\widetilde{\eta}^TK_p(.)\dot{\widetilde{\eta}}.
\end{equation}

Considering the fact that the AUV is moving at slow speed induced with zero wave frequency we have $\dot{M}(\eta)=0$ (for more details see \citep{Fos:02}, \citep{KeC:96}), since the matrix $\frac{1}{2}\dot{M}(\eta)-C(\nu,\eta)$ is skew symmetric and $D(\nu,\eta)\!>\!0$, then, by substituting the closed-loop tracking error equation \eqref{eq-1} in \eqref{lyap2}, we have
\begin{equation}\label{lyap3}
\dot{V}(\widetilde{\eta},\dot{\widetilde{\eta}})=-\dot{\widetilde{\eta}}^T\Big(D(\nu,\eta)+K_v(.)\Big)\dot{\widetilde{\eta}}+\dot{\widetilde{\eta}}^T\!\tau_w.
\end{equation}

Now, the $\cHi$ cost \citep{BEFB:94} is defined as follows
\begin{equation}\label{cost1}
J=\int_{0}^{\infty} (\dot{\widetilde{\eta}}^T\dot{\widetilde{\eta}}-\varepsilon^2 \tau_w^T\tau_w) \der t,
\end{equation}
with $\varepsilon>0$ is a positive number. Therefore, we have
\begin{equation}\label{cost2}
J<\int_{0}^{\infty} (\dot{\widetilde{\eta}}^T\dot{\widetilde{\eta}}-\varepsilon^2 \tau_w^T\tau_w +\dot{V}) \der t,
\end{equation}
and it follows that a sufficient condition for $J\leqslant0$ is that
\begin{equation}\label{equation-proof-thm-11}
\dot{V}+\dot{\widetilde{\eta}}^T\dot{\widetilde{\eta}}-\varepsilon^2 \tau_w^T\tau_w\leqslant0.
\end{equation}

Using \eqref{lyap3} and \eqref{equation-proof-thm-11}, the sufficient condition can be written as
\begin{equation}\label{equation-proof-thm-12}
\begin{bmatrix} \dot{\widetilde{\eta}} \\ \tau_w \end{bmatrix}^T\!\!
\begin{bmatrix} -\Big(D(\nu,\eta)+K_v(.)\Big) +I&  I\\  (\star)  & -\varepsilon^2 I \end{bmatrix}\!
\begin{bmatrix} \dot{\widetilde{\eta}} \\ \tau_w \end{bmatrix}\!\!<\!0,
\end{equation}
where $(\star)$ is used for the blocks induced by symmetry.

Thus, a sufficient condition for $J\leqslant0$ is that the following inequality be negative definite
\begin{equation}\label{equation-proof-thm-13a}
\begin{bmatrix} -\Big(D(\nu,\eta)+K_v(.)\Big) +I&  I\\  (\star)  & -\varepsilon^2 I \end{bmatrix}\!\!<\!0.
\end{equation}
Since the gain matrix $K_v(.)\!>\!0$ and the damping matrix satisfies $D(\nu,\eta)\!>\!0$ \citep{Fos:94}. Then, the Lyapunov function $\dot{V}$ is negative semidefinite. Finally, we conclude that the equilibrium point is asymptotically stable using the Krasovskii-Lasalle theorem \citep{KeC:96,CCCTL:17}, which ends the proof.
\end{proof}

It is worth noticing that if $\mu_{pj}\!=\!\mu_{vj}\!=\!1$, the proposed NLPD controller reduces and combines the PD controller. Hence, there is a close connection between the two proposed NLPD and PD controller's merits. Further, we propose a PD controller as a control tracking input given by
\begin{align}\label{eq-OW-a8}
 {\bf \tau}=R^T(\varrho) &\Big[ M(\eta)\ddot{\eta}_{\mbox{\scriptsize ref}}+C(\nu,\eta)\dot{\eta}_{\mbox{\scriptsize ref}}+D(\nu,\eta)\dot{\eta}_{\mbox{\scriptsize ref}}+K_p(.)\widetilde{\eta}+K_v(.)\widetilde{\nu}\Big],
\end{align}
where $\mu_{pj}\!=\!\mu_{vj}\!=\!1$, and $K_p$ and $K_v$ are the control gains of the PD controller.

\begin{remark}
Thanks to the additional tuning gains parameters, the NLPD controller is proved to offer superior tracking control performance \citep{CCCTL:17}. Additionally, the proposed NLPD controller is continuously updated by a sum of the position and velocity errors, reducing its computational complexity. 
\end{remark}

\section{Simulation results and discussions}\label{simulations}
We ran and interpreted the AUV-based trajectory tracking control simulations, which show the benefits of the two proposed control methods using MATLAB/Simulink environment for all of our simulations. The Simulink model contains both proposed PD and NLPD controllers and the AUV system dynamics. The sampling time and the end time simulation are set to $0.005${\si s} and $150${\si s}, respectively. The position of the surface ship and the AUV at each time is directly used to calculate the relative position error, which is used as error feedback with both proposed PD and NLPD controllers. As we also are dealing with a communication link performance, we follow the performance metrics illustrated in Fig. \ref{fig-Metrics} in which the parameters are summarized in Table \ref{para1}.

\begin{figure}[!t]
\centering
      \begin{overpic}[scale=0.35]{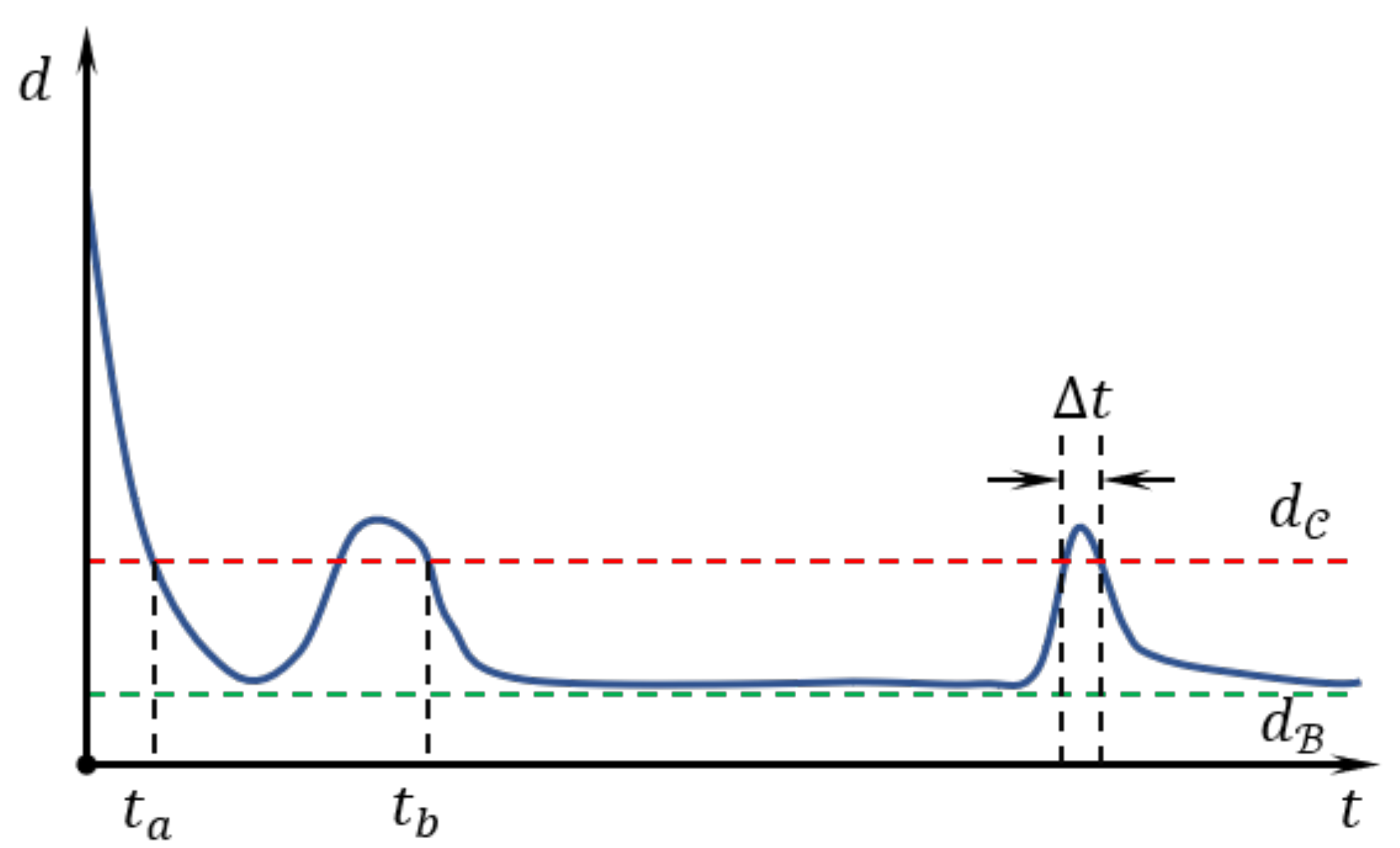}
      \put(45,-1.5){ \scriptsize Time [s]}
        \put(1,18){\scriptsize \begin{rotate}{90} Distance~[{\si m}] \end{rotate}}
            \end{overpic}  
              \caption{ Example that illustrates the stability of the closed-loop
control system and the optical communication performance metrics.}\label{fig-Metrics} 
\end{figure}

\begin{table}[!t]
\caption{Parameters described in Fig. \ref{fig-Metrics}.}
\begin{center}
\begin{tabular}{| c || c | }
  \hline 
\textbf{~~Symbol~~}  & \textbf{~~~Description~~~}   \\ \hline  
$d$ &  ~~~Distance between AUV and the surface ship~~~\\  
~~$\maxd$~~ & Slant height of the cone\\  
~~~$\mind$~~~ &Minimum possible distance for an AUV \\  
 &working at a certain depth $h$, $\mind = h$ \\  \hline  
\end{tabular}
\end{center}
\label{para1}
\end{table}

{\emph{Cone Arrival Time $t_a$}}: The time when the AUV enters the cone \emph{for the first time} and the distance is less than the slant height of the cone-shaped beam region $d_{\cal C}$ is defined as cone arrival time. Once the AUV enters the connectivity cone, the ship's receiver will observe a continuous spot on its wide field of view position detector identifying the AUV transmitter position. Then, the AUV can establish communication and upload images and other data.

{\emph{Communication Established Time $t_b$}}: Ignoring the effect of external disturbances, the time when AUV enters the cone and \emph{stays inside after} is defined as the communication established time.

{\emph{Communication Restoring Time $\Delta t$}}: Suppose an impulse input disturbance $\tau_w$ \emph{which might lead the AUV to out-of-cone status} is introduced into the system. The time spent from the moment that the AUV gets out of the cone to the moment when AUV returns to the cone and \emph{stays inside after} is defined as restoring time.

{\emph{Root Mean Square Error $\mathrm{RMSE}_\chi$}}:  The $\mathrm{RMSE}_\chi$ values quantify the effect of the closed-loop steady state error performance after communication for a given interval of time with respect to the proposed controllers and permit us to conclude on the performance results.

To test the ability of the two proposed target tracking controllers for underwater based optical communication through numerical simulations, we set the initial position of the transmitter at $\begin{bmatrix} 5 & 5 & 0  \end{bmatrix}^T$. At the same time, the mobile ship receiver is at the origin with a heading angle $\varrho_{\mbox{\tiny RX}}=0^\circ$. Table \ref{sample-table-AUV} gives the parameter of the AUV dynamics, which are adapted from \citep{CGEC:10}.  The parameters $\mu_{pj}$ and $\mu_{vj}$ are tuned manually to improve the system behavior. For more details on selecting the parameters $\mu_{pj}$ and $\mu_{vj}$, the interested reader is referred to \citep{CCCTL:17}. The transmitter-receiver distance is approximately $8.03${\si m} at the beginning and the controller's performance gains $K_p\!=\!300$ and $K_v\!=\!250$ are chosen appropriately to achieve a good compromise between trajectory tracking and optical communication performances.

Let the ship be at $\begin{bmatrix} 0&10&0 \end{bmatrix}$ when the AUV starts to track the ship and we make that the mobile ship receiver follows a circular trajectory described by the following circular motion. 
\begin{equation*}
\begin{bmatrix} x_{\mbox{\tiny RX}}(t)\\ y_{\mbox{\tiny RX}}(t)\\ \varrho_{\mbox{\tiny RX}}(t) \end{bmatrix}=
\begin{bmatrix} 10 \sin (0.1t)\\ 10 \sin (0.1t + \frac{\pi}{2}) \\ \frac{\pi}{2}\sin(0.1t)\end{bmatrix}.
\end{equation*}

\subsection{Target Tracking Performance in Nominal Conditions}\label{sub-sec1}
To force the AUV transmitter to reach the cone-shaped beam region and then stay within it without considering any external disturbances, we use the feedback controllers \eqref{eq-OW-aa} and \eqref{eq-OW-a8} to control the AUV. We make that the surface ship vehicle follows a circular trajectory. Then, the trajectories to time and $x$--$y$--$\varrho$ plane described by the receiver and the transmitter using PD and NLPD controllers without any external disturbances are illustrated in Figs. \ref{fig-Traj1}\textbf{a)} and \ref{fig-Traj1}\textbf{b)}, respectively. The red dashed plot and the blue curve are the simulated AUV trajectories using PD and NLPD controllers, respectively. The green curve is the surface ship trajectories that follow a desired circular path. Both PD and NLPD controllers move the AUV transmitter to the desired trajectory, which guarantees the closed-loop stability. Table \ref{tab:rmse_pd_fd_nominal} shows that the NLPD controller reduces $x$ and $y$ axis states tracking errors by more than $70$\% from the states tracking errors of the PD controller. Subsequently, the NLPD controller generates better convergence error than the PD controller.  

Since the link communication quality is determined by the distance between the surface ship vehicle and the AUV, then the different effect of the angle $\varrho$ using the NLPD controller to the PD controller can be neglected from a practical point of view. To test the two proposed controllers' ability, we plot the distance $d$ to the transmitter-receiver line, the logarithm of the data rates $ B $, and the receiver pointing error $\psi$ as shown in Figs. \ref{fig-nominal}\textbf{a)}, \ref{fig-nominal}\textbf{b)}, and \ref{fig-angle1}, respectively. We note that after around $t_a=2$ {\si s}, the distance of both PD and NLPD controllers is less than $d_{\cal C}$ and stays less this threshold. Meanwhile, the AUV remains in the cone-shaped beam region, and its bit rate is guaranteed to be around $10$ {\si Mbps}. 

\begin{table*}[!t]
\caption{Modeling parameters of the AUV.}
\begin{center}
\begin{tabular}{| c | c |c |}
  \hline
~$m_{11}=100$ {\si kg}~  & ~~~~~$m_{22}=250$ {\si kg}~~~~  & ~~$m_{33}=80$ {\si kg}~~      \\ \hline
~~$d_{11}=\!(70\!+\!100\abs{u})${\si kg/s}~~&~~$d_{22}\!=\!(100\!+\!200\abs{v})${\si kg/s}~~ & ~~$d_{33}\!=\!(50\!+\!100\abs{r})${\si kg.m$^2$/s}~~    \\ \hline
\end{tabular}
\end{center}
\label{sample-table-AUV}
\end{table*}

\begin{table}[!t]
    \centering
        \caption{Performance evaluation between PD and NLPD controllers in nominal conditions.}
    \begin{tabular}{ccccc}
    \toprule			
~Control scheme~	&	$~\mathrm{RMSE}_x$[{\si m}]~	&	~$\mathrm{RMSE}_y$[{\si m}]~	&	~$\mathrm{RMSE}_\varrho$[{\si rad}]~	&	~$t_a$[{\si s}]~	\\	\midrule
PD controller	&	0.4583	&	0.4557	&	0.0230	&	2.015	\\	
NLPD controller	&	0.0596	&	0.1197	&	0.0408	&	2.315	\\	\midrule
Improvement	&	\textbf{87.00\%}	&	\textbf{73.73\%}	& -77.39\%	&	-0.3	\\	\bottomrule
\end{tabular}
    \label{tab:rmse_pd_fd_nominal}
\end{table}

\begin{figure*}[!t]
   \begin{minipage}[c]{0.45\linewidth}  
           \centering
      \begin{overpic}[scale=0.26]{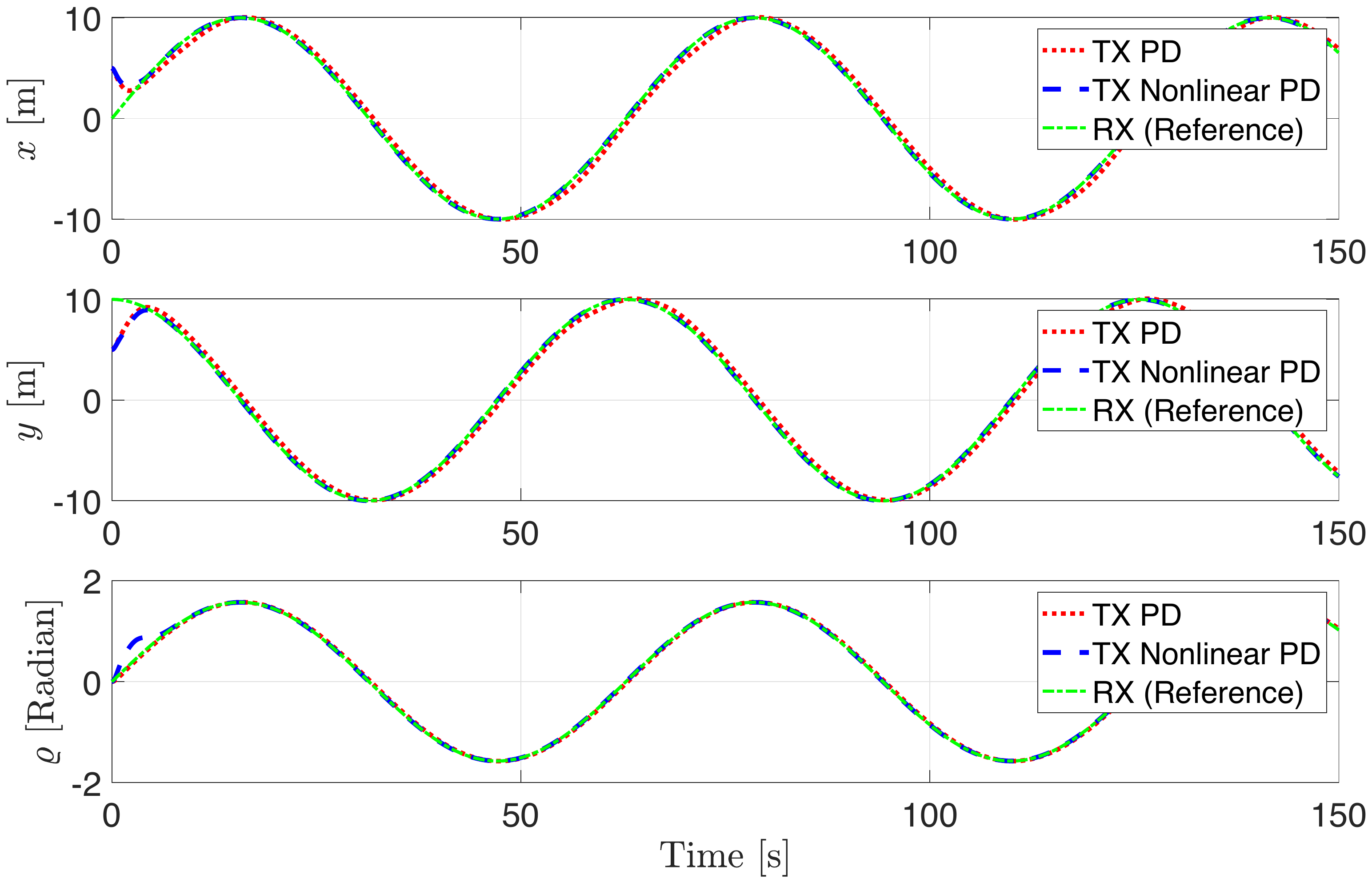}
      \put(20,0){\scriptsize  \textbf{a)}}
            \end{overpic}  
       \end{minipage}\hfill 
         \begin{minipage}[c]{0.48\linewidth}
         \centering
      \begin{overpic}[scale=0.26]{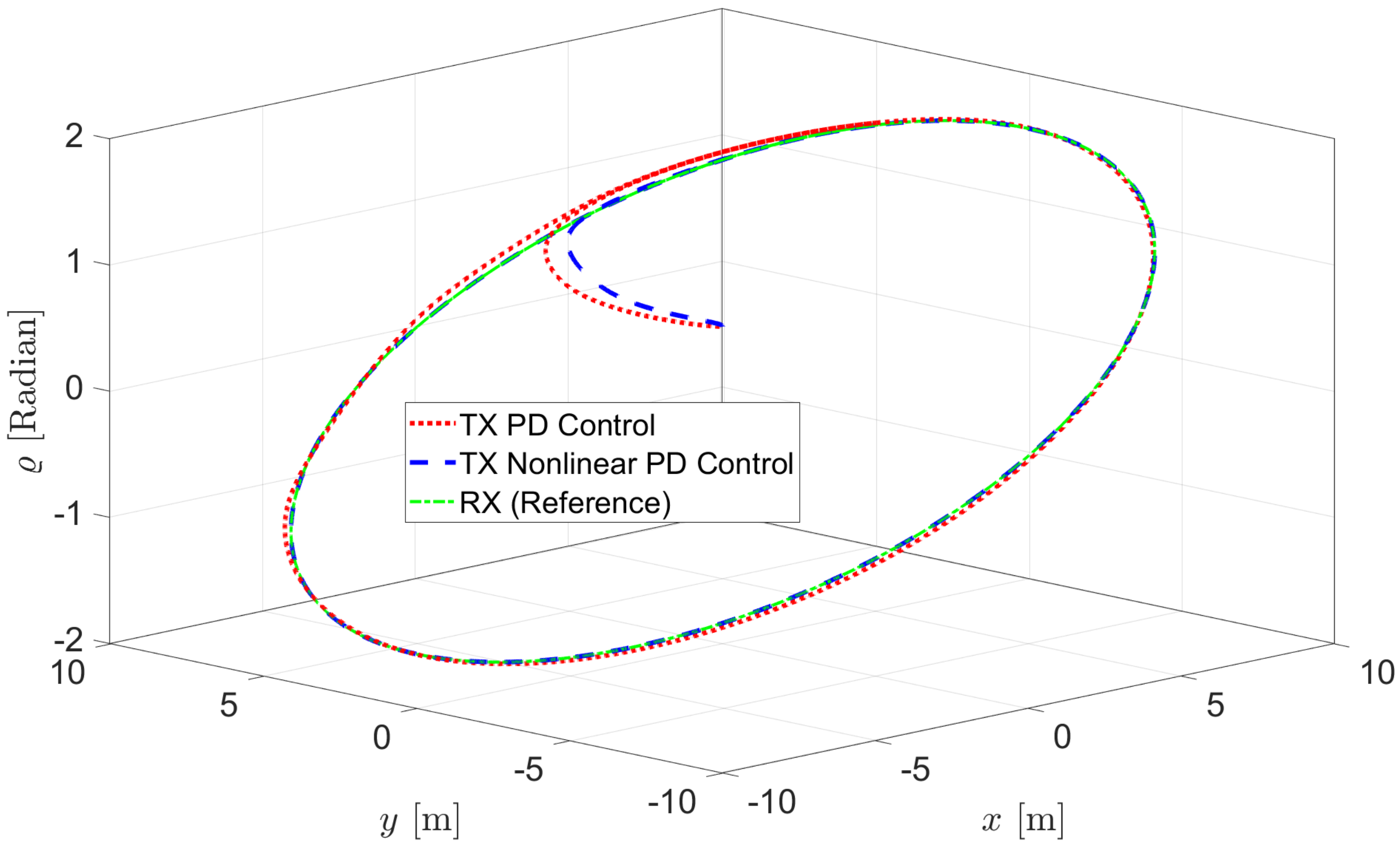}
      \put(20,-1){ \scriptsize \textbf{b)}}
            \end{overpic}  
   \end{minipage}         
      \caption{\textbf{a)} States responses described by the surface ship receiver and the AUV transmitter using PD and NLPD controllers without disturbances; \textbf{b)} 3D-view.} \label{fig-Traj1}
\end{figure*}

\begin{figure*}[!t]
   \begin{minipage}[c]{0.45\linewidth} 
   \centering
      \begin{overpic}[scale=0.25]{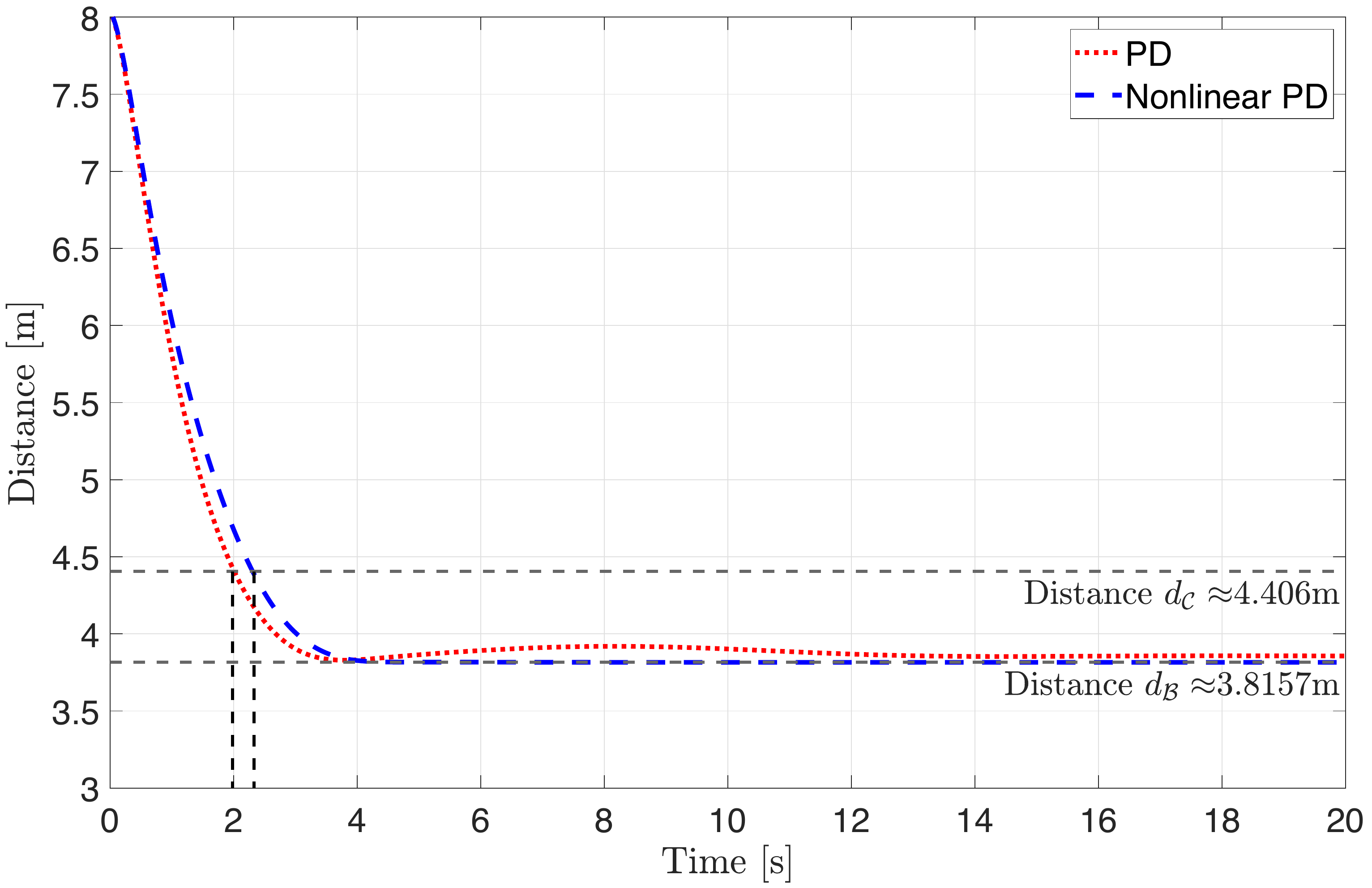}
       \put(20,1){ \scriptsize \textbf{a)}}
            \end{overpic}  
       \end{minipage}\hfill 
         \begin{minipage}[c]{0.48\linewidth}
         \centering
      \begin{overpic}[scale=0.25]{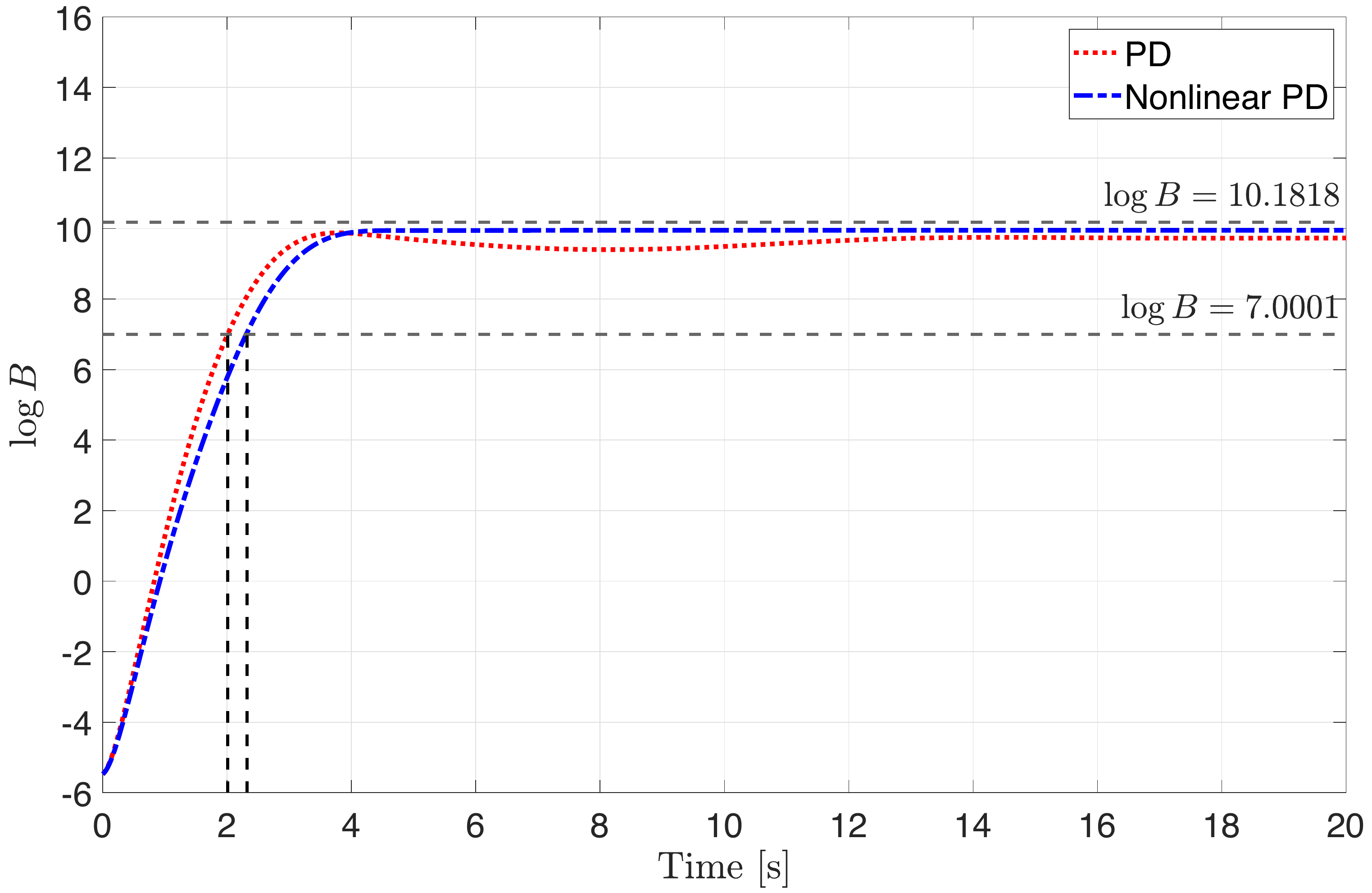}
      \put(20,0){\scriptsize \textbf{b)}}
            \end{overpic}  
   \end{minipage}         
      \caption{NLPD controller versus PD controller: \textbf{a)} Distance $d$ between transmitter-receiver channel; \textbf{b)} Logarithm of the bit rate.} \label{fig-nominal}
\end{figure*}

\begin{figure*}[!t]
\centering
         \begin{minipage}[c]{0.48\linewidth}
      \begin{overpic}[scale=0.25]{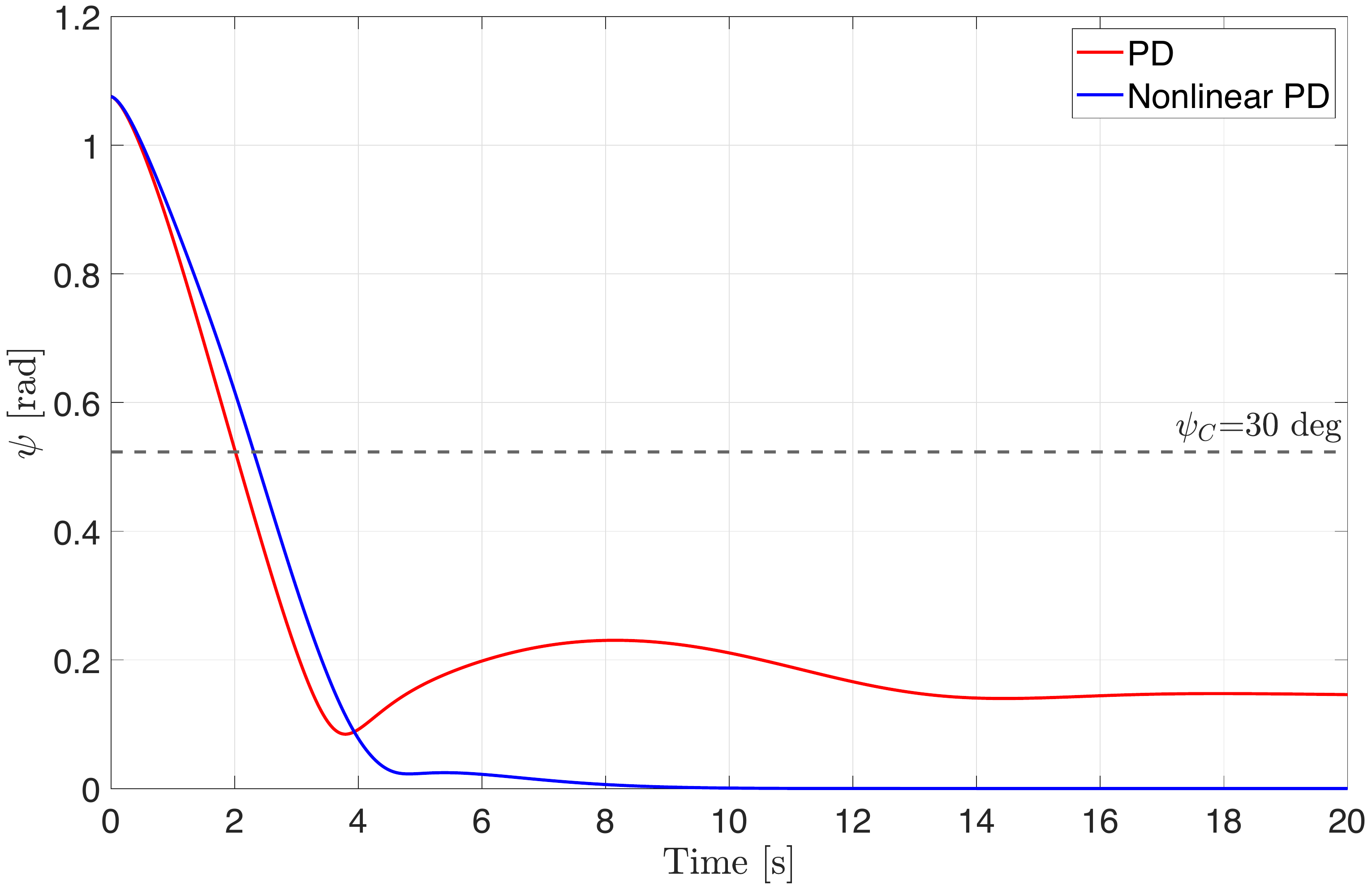}
      \put(20,0){\scriptsize \textbf{c)}}
            \end{overpic}  
   \end{minipage}         
      \caption{NLPD controller versus PD controller: Receiver pointing error $\psi$.} \label{fig-angle1}
\end{figure*}

\subsection{Robust Test}\label{sub-sec2}
The inherent robustness to the external disturbances and measurement noises is an essential factor for marine control systems. Therefore, we investigate the robustness of the tracking control in underwater based optical communication. To show the PD and NLPD controllers' robustness for maintaining a perfect position between transmitter and receiver to establish a directed optical LoS link, we simulate the AUV tracking control under two strict test conditions.  

\subsubsection{Case I: Robustness toward ocean current forces and disturbances}\label{case1}
In the first case, current ocean forces of magnitude $\begin{bmatrix} 350~({\si N}) & 350~({\si N})& 350~({\si N.m}) \end{bmatrix}^T$ and measurement noises are introduced as disturbances. The current ocean disturbances occur at $30$ {\si s} and last for $1$ {\si s}.

\begin{figure*}[!t]
      \vspace{0.5cm}
   \begin{minipage}[c]{0.45\linewidth}  
           \centering
      \begin{overpic}[scale=0.26]{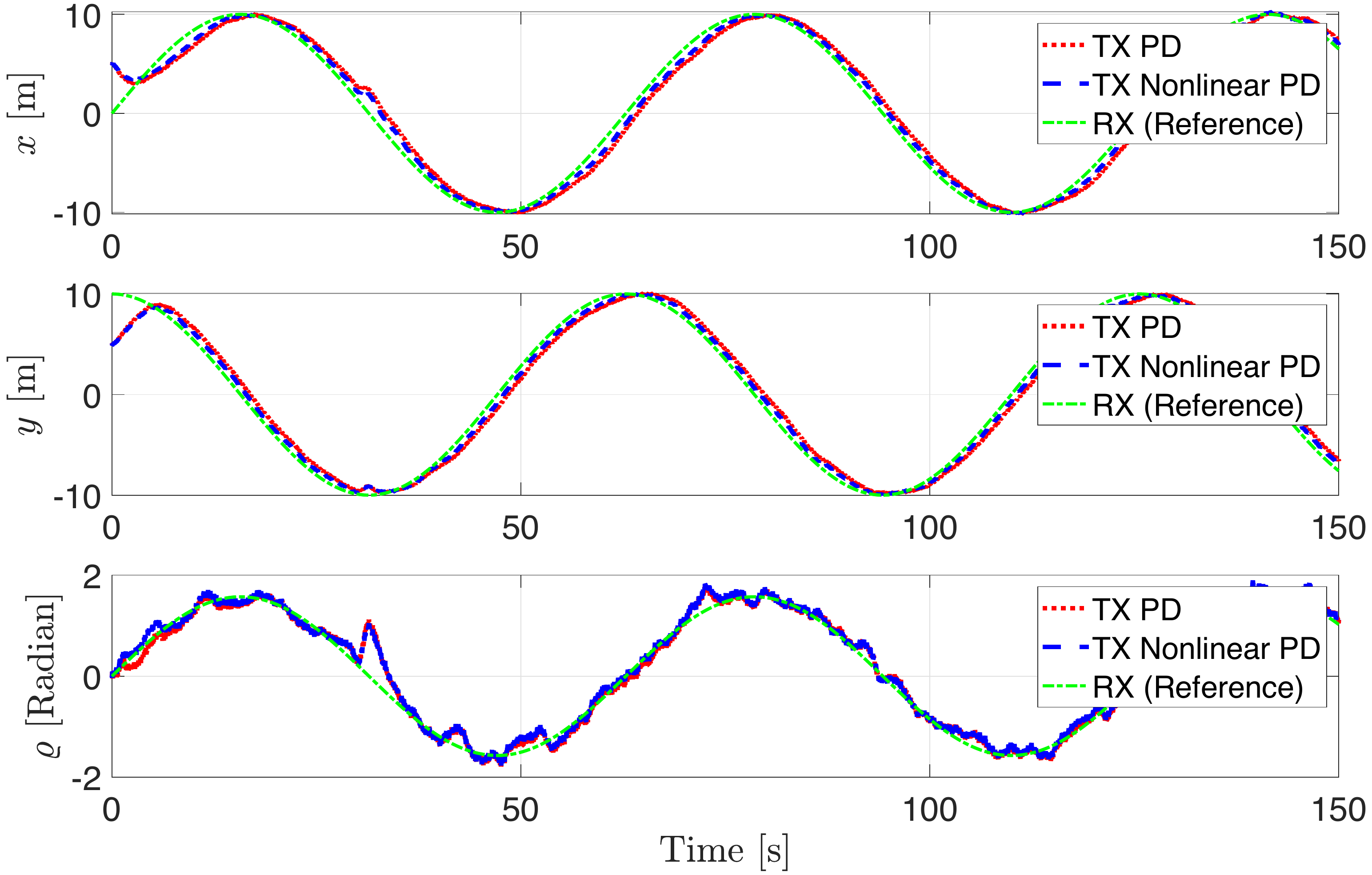}
      \put(20,0){\scriptsize  \textbf{a)}}
            \end{overpic}  
       \end{minipage}\hfill 
         \begin{minipage}[c]{0.48\linewidth}
         \centering
      \begin{overpic}[scale=0.26]{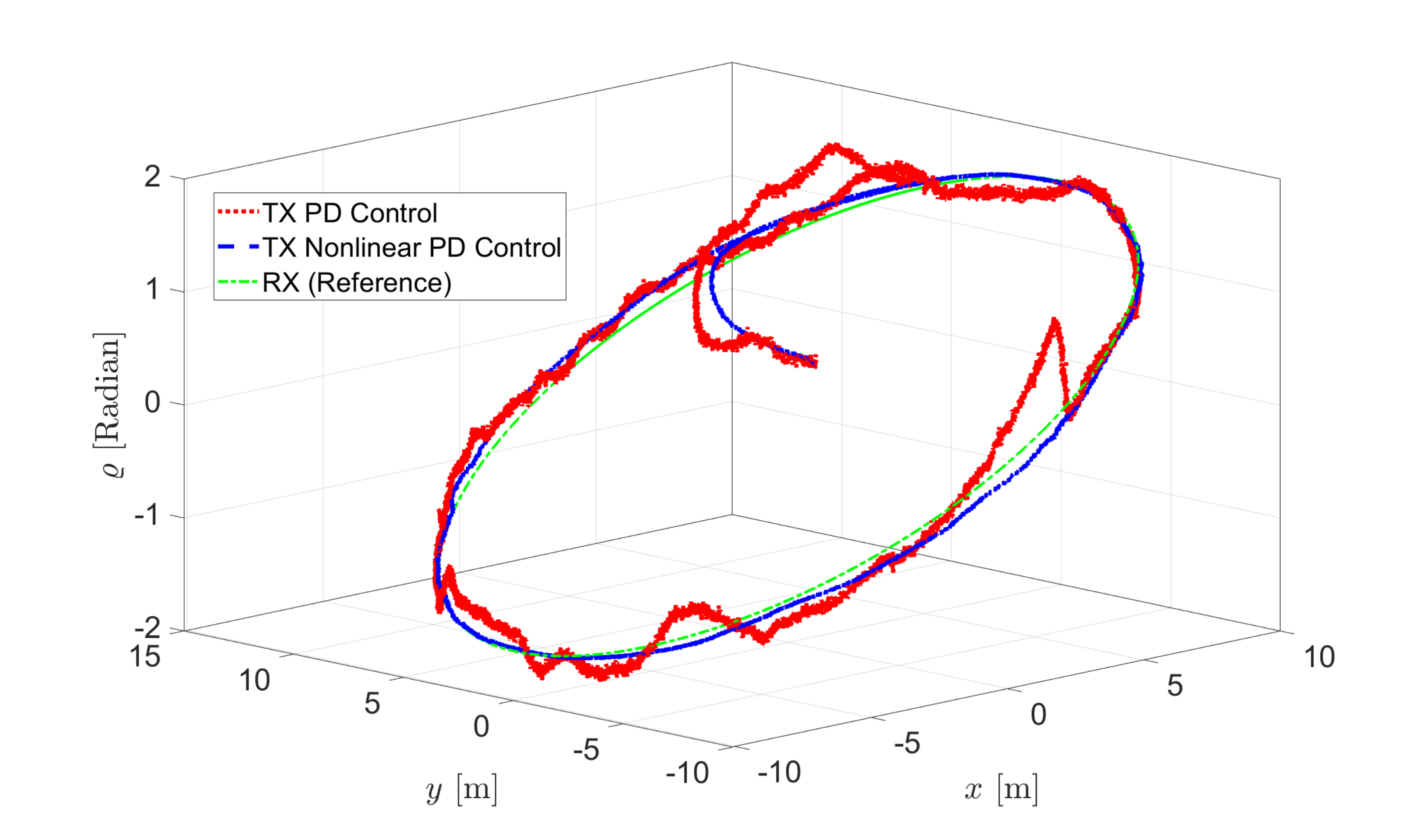}
      \put(20,-1){ \scriptsize \textbf{b)}}
            \end{overpic}  
   \end{minipage}         
      \caption{\textbf{a)} States responses described by the AUV transmitter and the surface ship receiver using PD and NLPD controllers with disturbances---Case I; \textbf{b)} 3D-view---Case I.} \label{fig-Rob1}
\end{figure*}

\begin{figure*}[!t]
   \begin{minipage}[c]{0.45\linewidth} 
   \centering
      \begin{overpic}[scale=0.25]{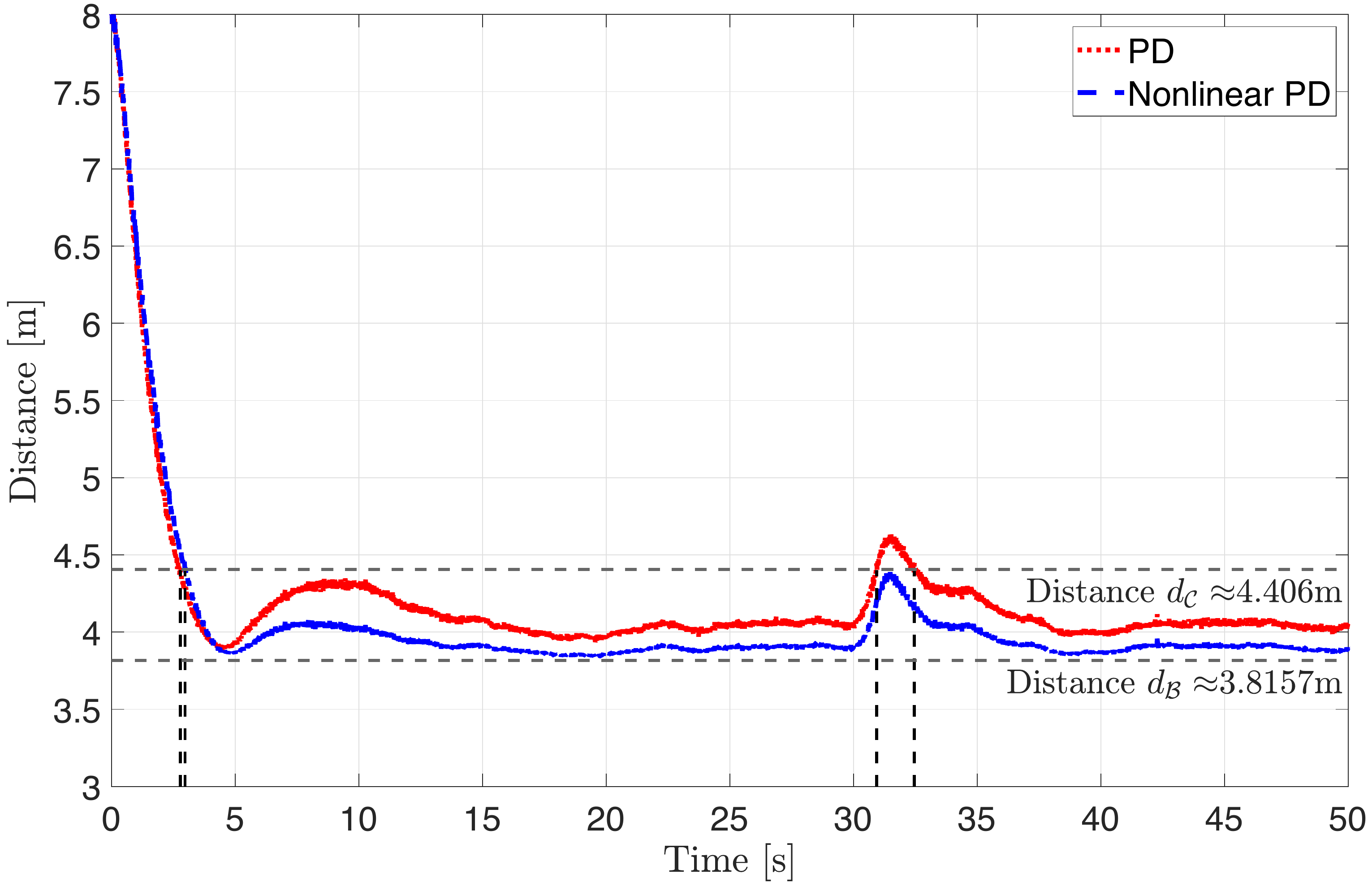}
               \put(64,26){\scriptsize $\downarrow$}
        \put(42,29){ \tiny  Data transmission loss with PD controller}
       \put(20,1){ \scriptsize \textbf{a)}}
            \end{overpic}  
       \end{minipage}\hfill 
         \begin{minipage}[c]{0.48\linewidth}
         \centering
      \begin{overpic}[scale=0.25]{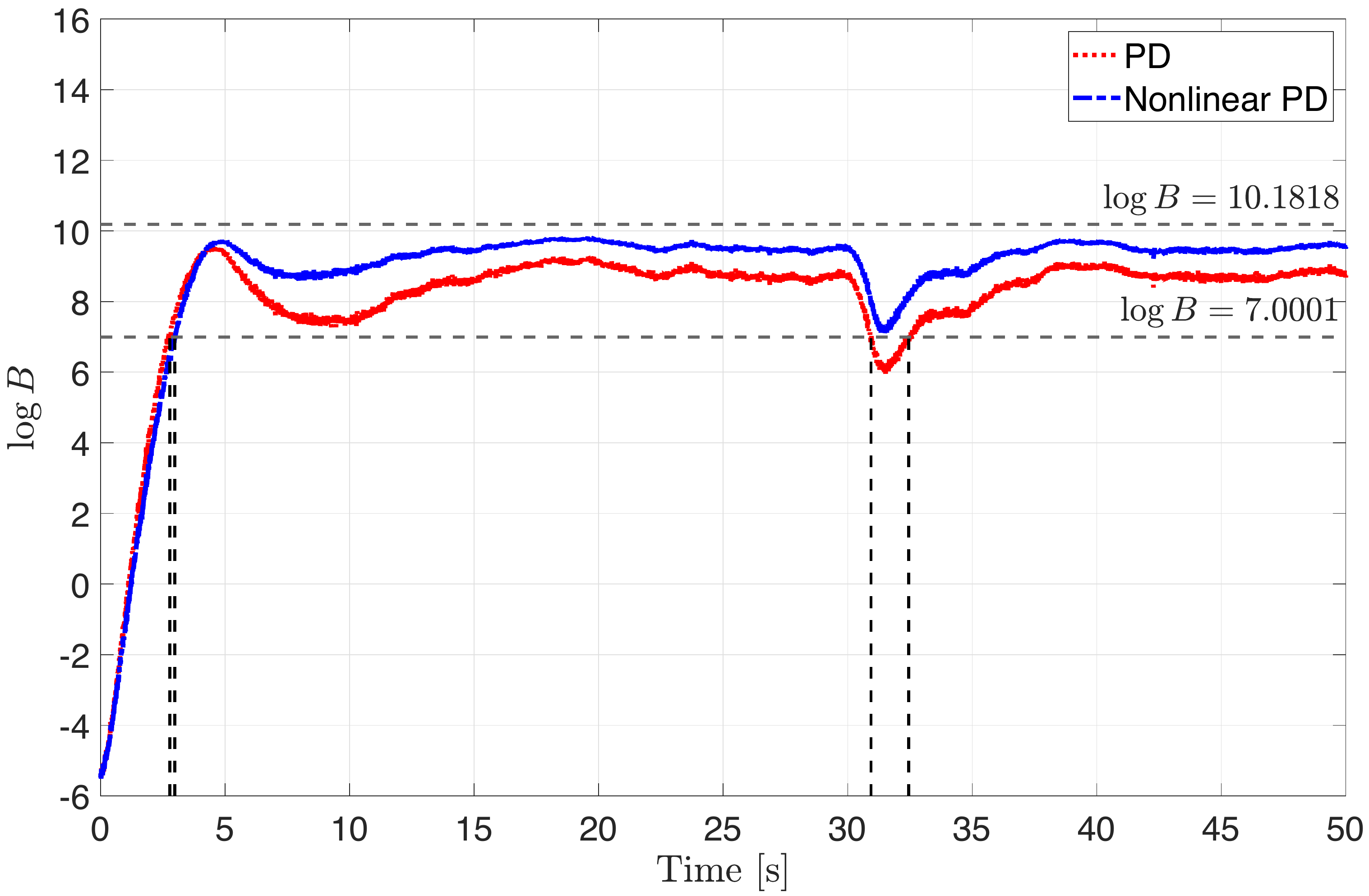}
                                      \put(64,36){\scriptsize $\uparrow$}
        \put(42,33){ \tiny  Data transmission loss with PD controller}
      \put(20,0){\scriptsize \textbf{b)}}
            \end{overpic}  
   \end{minipage}         
      \caption{NLPD controller versus PD controller with measurement noises and current ocean forces---Case I: \textbf{a)} Distance $d$ between transmitter and receiver channel; \textbf{b)} Logarithm of the bit rate.} \label{fig-robust}
\end{figure*}

\begin{figure*}[!t]
\centering
         \begin{minipage}[c]{0.48\linewidth}
      \begin{overpic}[scale=0.25]{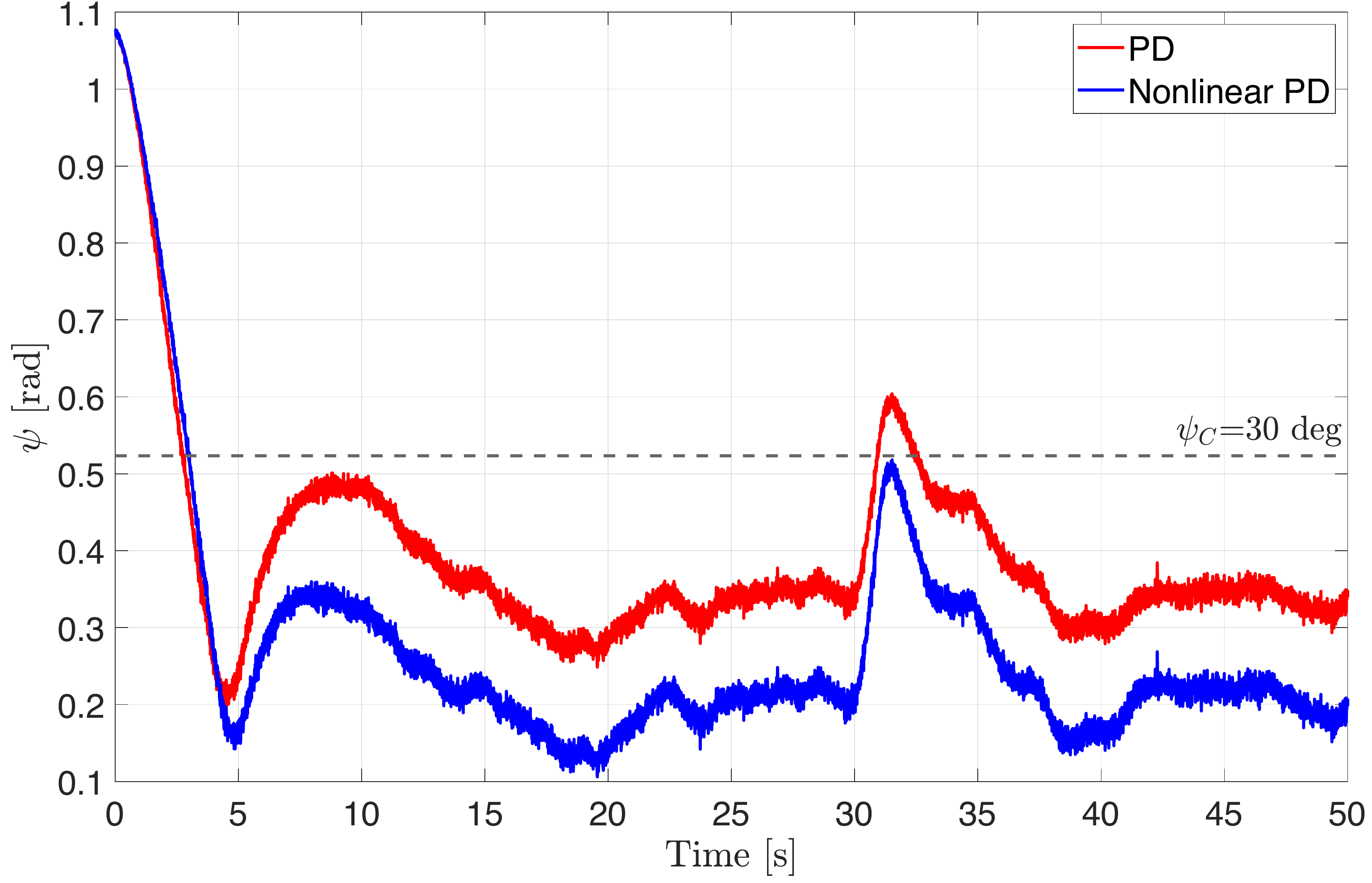}
                                \put(64,36){\scriptsize $\downarrow$}
        \put(42,39){ \tiny  Data transmission loss with PD controller}
            \end{overpic}  
   \end{minipage}         
      \caption{NLPD controller versus PD controller with measurement noises and current ocean forces---Case I: Receiver pointing error $\psi$.} \label{fig-angle_disturbance}
\end{figure*}

\begin{table}[!t]
    \caption{Performance evaluation between PD and NLPD controllers with current ocean forces and measurement noises---Case I.}\label{tab:rmse_pd_fd_ed_nm}
    \centering
  \begin{tabular}{cccccc}
\toprule										
Control scheme	&	$\mathrm{ RMSE}_x$[{\si m}]	&	$\mathrm{RMSE}_y$[{\si m}]	&	$\mathrm{RMSE}_\varrho$[{\si rad}]	&	~$t_a$[{\si s}]~	&	~$\Delta t$[s]~	\\	\midrule
PD controller	&	1.103	&	0.998	&	0.187	&	2.780	&	1.08	\\	
NLPD controller	&	0.687	&	0.597	&	0.185	&	2.975	&	--	\\	\midrule
Improvement	&	\textbf{37.72\%}	&	\textbf{40.18\%}	&	-1.07\%	&	-0.195	&	-	\\	\bottomrule
\end{tabular}
\end{table}

From the simulation results illustrated in Figs. \ref{fig-Rob1}\textbf{a)} and \ref{fig-Rob1}\textbf{b)} with disturbances, we find that the NLPD based tracking control still leads the AUV well converged to the desired trajectory, while the tracking control with PD controller exhibits large tracking errors. The RMSE for both PD and NLPD controllers are summarized in Table \ref{tab:rmse_pd_fd_ed_nm}. Roughly, the NLPD controller reduces $x$ and $y$ axis states tracking errors by more than $35$\% from the states tracking errors of the PD controller. Subsequently, the NLPD controller generates better robustness performance errors and has significant improvement in $\mathrm{RMSE}$ in $x$ and $y$ direction while the PD controller lacks such an advantage.

To test the proposed controllers' performance with measurement noises and current ocean forces, we plot the distance $d$ to transmitter-receiver line, the logarithm of the data rates $ B $ and the receiver pointing error $\psi$ as shown in Figs. \ref{fig-robust}\textbf{a)}, \ref{fig-robust}\textbf{b)}, and \ref{fig-angle_disturbance}, respectively. We note that after around $t_a=2.5${\si s} and $t_a=3.015${\si s} the distance of the PD and NLPD controllers is less than $d_{\cal C}$ and settles less than this value all along the simulation. Then, the AUV stays within the cone-shaped beam region from this time on and the bit rate is guarantee to be around $10$ {\si Mbps} as shown in Fig. \ref{fig-robust}\textbf{b)}. Subsequently, when a short external input disturbance is applied to the system during communication link, NLPD controller has a much smaller $\Delta t$ and overshoot while PD controller lacks to maintain the communication link.

\subsubsection{Case II: Robustness toward parameter's uncertainties, current ocean forces, and disturbances}\label{case2}
To assess the robustness of the two proposed control schemes against model variation, we increase mass parameter by $20\%$, ocean current disturbances of magnitude $\begin{bmatrix} 350~({\si N}) & 350~({\si N})& 350~({\si N.m}) \end{bmatrix}^T$ and measurement noises in the second case.

\begin{figure*}[!t]
   \begin{minipage}[c]{0.45\linewidth}  
           \centering
      \begin{overpic}[scale=0.26]{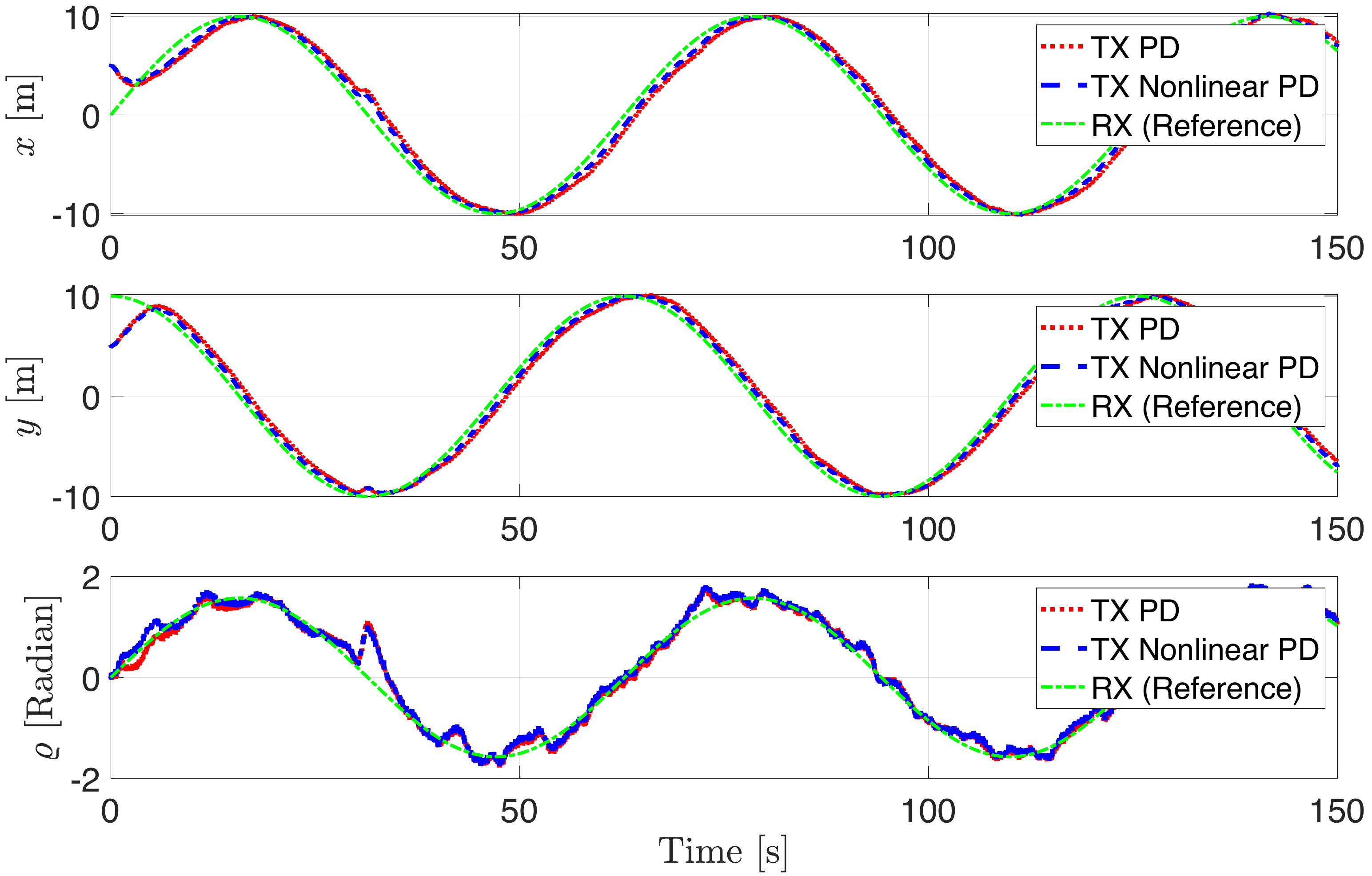}
      \put(20,0){\scriptsize  \textbf{a)}}
            \end{overpic}  
       \end{minipage}\hfill 
         \begin{minipage}[c]{0.48\linewidth}
         \centering
      \begin{overpic}[scale=0.26]{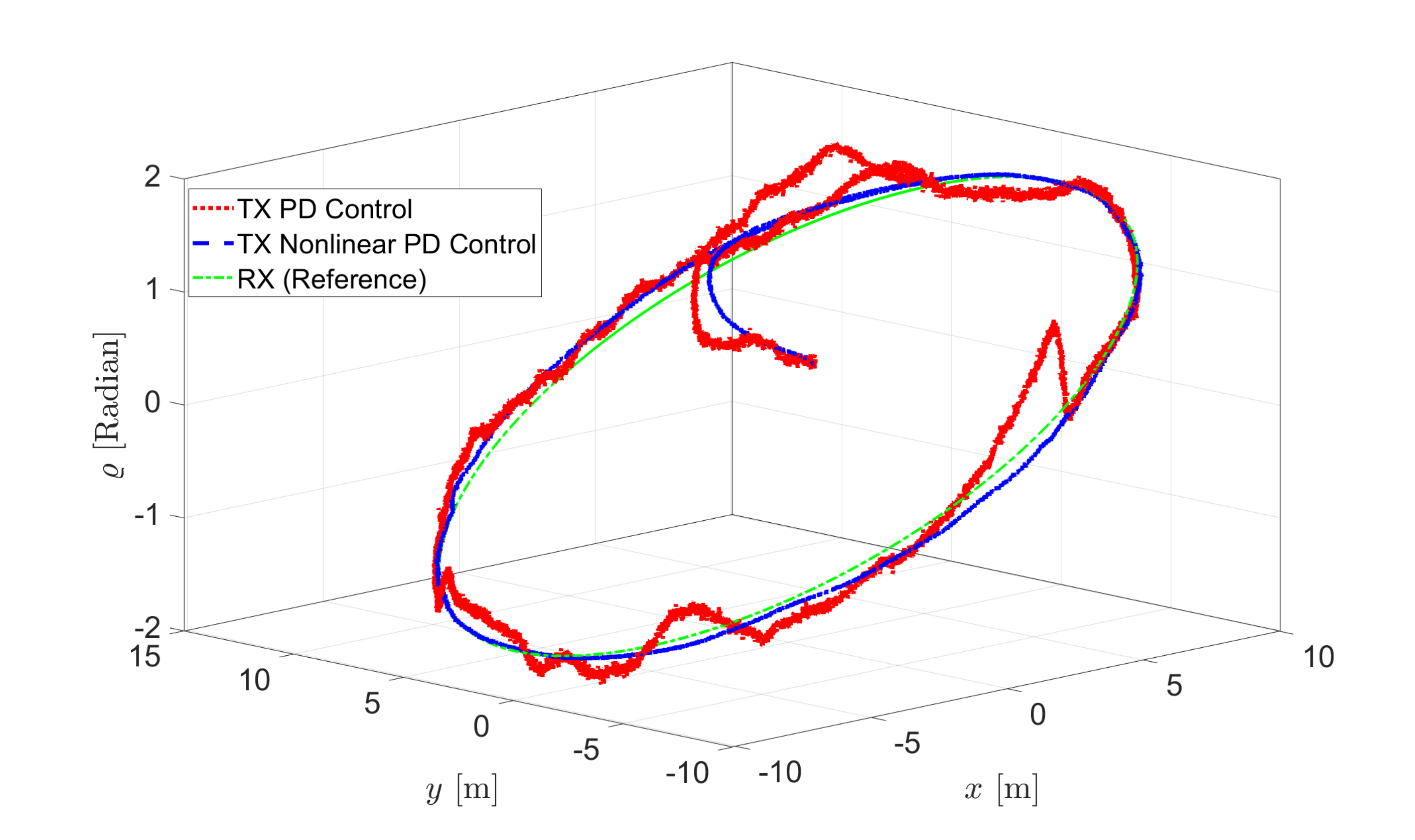}
      \put(20,-1){ \scriptsize \textbf{b)}}
            \end{overpic}  
   \end{minipage}         
      \caption{\textbf{a)} States responses described by the AUV transmitter and the surface ship receiver using PD and NLPD controllers with disturbances---Case II; \textbf{b)} 3D-view---Case II.} \label{fig-Rob2}
\end{figure*}

\begin{figure*}[!t]
   \begin{minipage}[c]{0.45\linewidth} 
   \centering
      \begin{overpic}[scale=0.24]{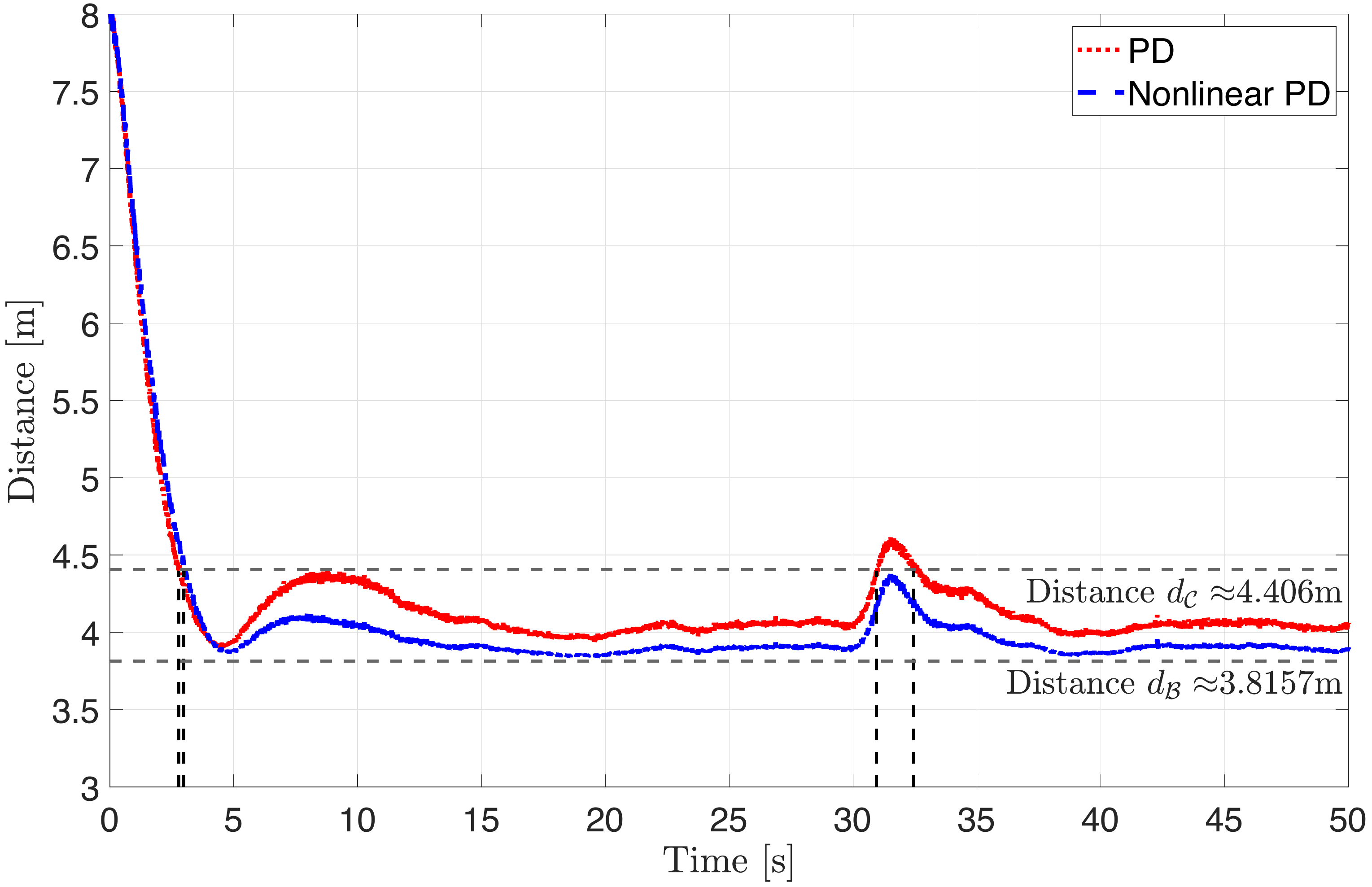}
                                \put(64,26){\scriptsize $\downarrow$}
        \put(42,29){ \tiny  Data transmission loss with PD controller}
        \put(22.5,27.5){\scriptsize \Bigg\downarrow}
        \put(11.5,35){ \tiny \textcolor{black}{\textbf{Very close to the distance limit between transmitter-receiver}}}
       \put(20,1){ \scriptsize \textbf{a)}}
            \end{overpic}  
       \end{minipage}\hfill 
         \begin{minipage}[c]{0.48\linewidth}
         \centering
      \begin{overpic}[scale=0.24]{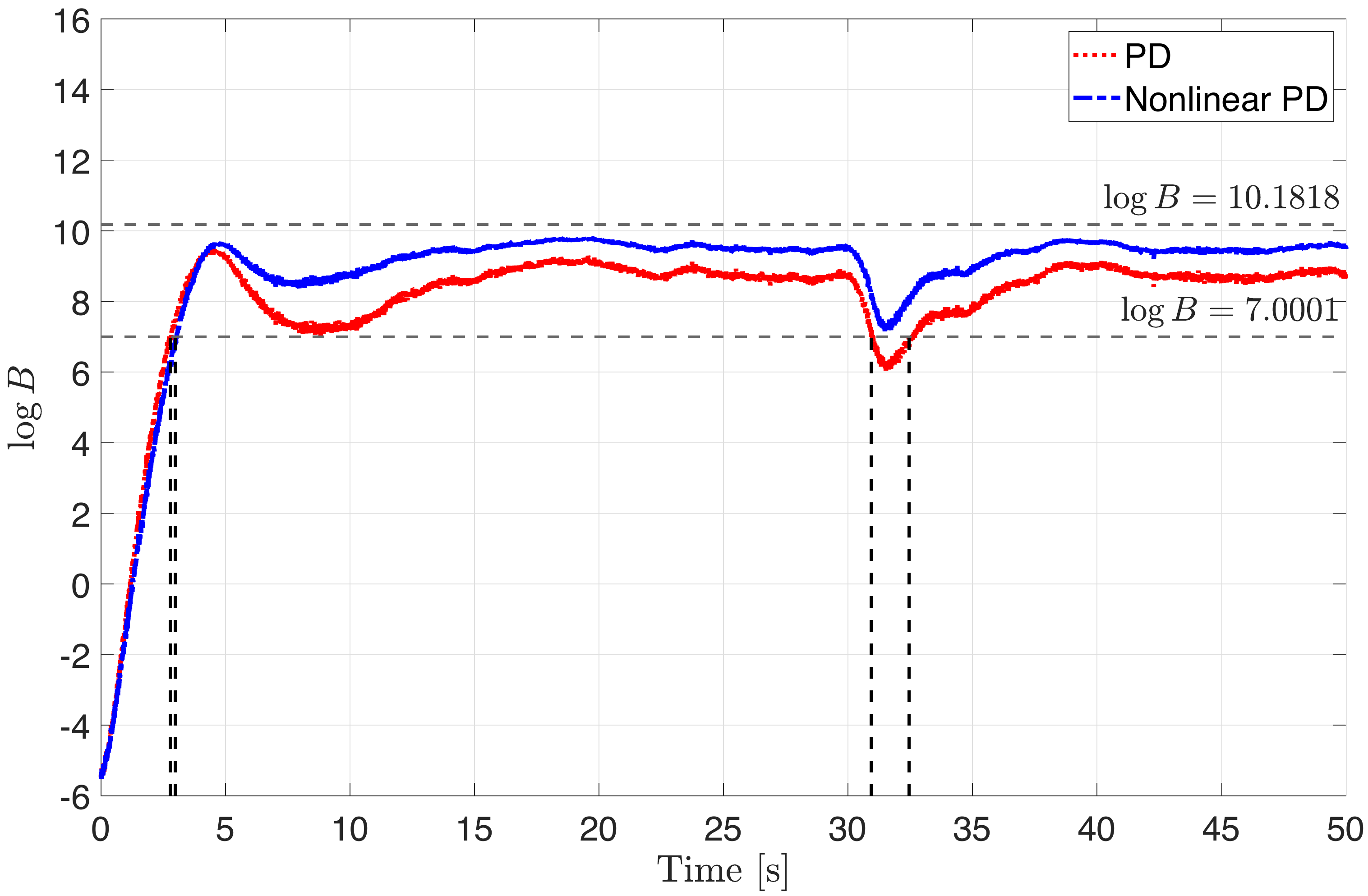}
                                \put(64,36){\scriptsize $\uparrow$}
        \put(42,33){ \tiny  Data transmission loss with PD controller}
                    \put(22,35){\scriptsize \Bigg\uparrow}
        \put(12,29){ \tiny \textcolor{black}{\textbf{Threshold limit of the bit rate}}}
      \put(20,0){\scriptsize \textbf{b)}}
            \end{overpic}  
   \end{minipage}         
      \caption{NLPD controller versus PD controller with with $20\%$ of mass parameter errors, measurement noises, and current ocean  forces---Case II: \textbf{a)} Distance $d$ between transmitter-receiver channel; \textbf{b)} Logarithm of the bit rate.} \label{fig-robust2}
\end{figure*}

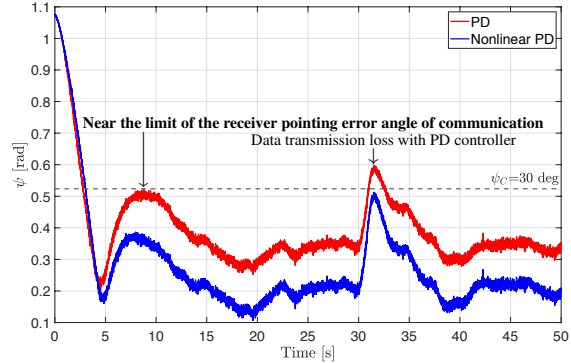
\begin{figure*}[!t]
\centering
         \begin{minipage}[c]{0.48\linewidth}
      \begin{overpic}[scale=0.24]{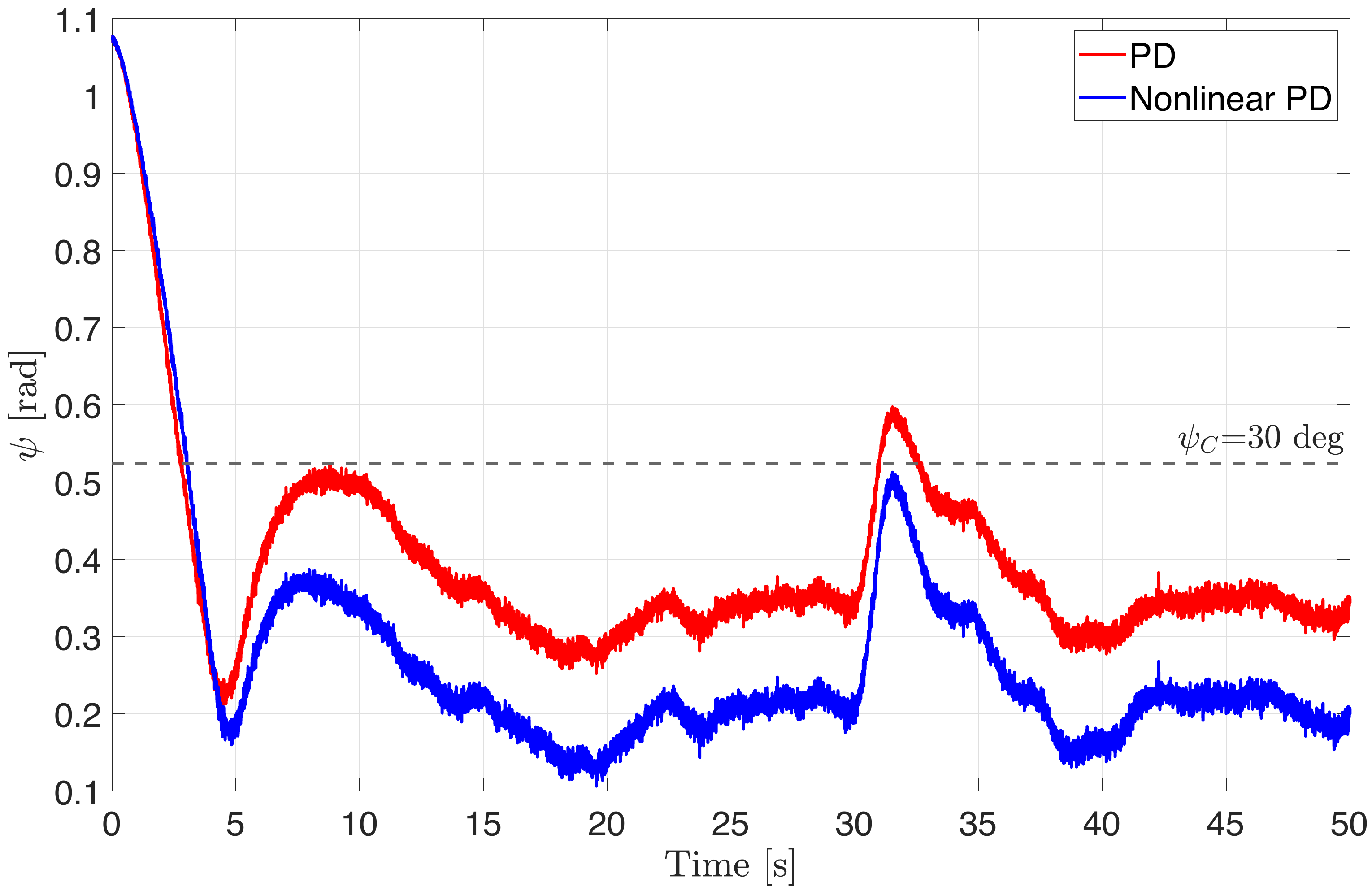}
                          \put(64,36){\scriptsize $\downarrow$}
        \put(42,39){ \tiny  Data transmission loss with PD controller}
              \put(22.75,35.5){\scriptsize \Bigg\downarrow}
        \put(12,42){ \tiny \textcolor{black}{\textbf{Near the limit of the receiver pointing error angle of communication}}}
            \end{overpic}  
   \end{minipage}         
      \caption{NLPD controller versus PD controller with $20\%$ of mass parameter errors, measurement noises, and current ocean forces---Case II: Receiver pointing error $\psi$.} \label{fig-angle_disturbance2}
\end{figure*}

Figs. \ref{fig-Rob2}\textbf{a)} and \ref{fig-Rob2}\textbf{b)} illustrate the performance of the two proposed controllers with disturbances, we observe that the responses of the NLPD based tracking control still lead the AUV well converged to the desired trajectory and guaranteed connectivity. In contrast, the tracking control with the PD controller exhibits significant tracking errors close to the optical communication link's threshold limit.

Further comparative details with the RMSE for both PD and NLPD controllers are summarized in Table \ref{tab:rmse_para}. Roughly, the NLPD controller reduces $x$ and $y$ axis states tracking errors by more than $35$\% from the states tracking errors of the PD controller. Subsequently, the NLPD controller generates better robustness performance errors and has significant improvement in $\mathrm{RMSE}$ in $x$ and $y$ direction while the PD controller lacks such an advantage.

Figs. \ref{fig-robust2}\textbf{a)}, \ref{fig-robust2}\textbf{b)}, and \ref{fig-angle_disturbance2} show the distance $d$ with respect to the transmitter-receiver line, the logarithm of the data rates $B$ and the receiver pointing error $\psi$, respectively. We observe that after around $t_a=2.5$ {\si s} and $t_a=3.015$ {\si s} the distance of the PD and NLPD controllers is less than $d_{\cal C}$ and hovers less this threshold all along the simulation. Hence, the AUV settles within the cone-shaped beam region at the same time and the bit rate is guarantee to be around $10$ {\si Mbps} as shown in Fig. \ref{fig-robust2}\textbf{b)}. Subsequently, when a short ocean current force is applied to the system during communication link, NLPD controller has a much smaller $\Delta t$ and overshoot while PD controller fails to maintain the optical communication link. 

Finally, our objective of providing and maintaining an optical communication link and a positioning tracking performance between the surface ship receiver and the AUV transmitter with an average bit rate of $10$ {\si Mbps} is well accomplished using both PD and NLPD controllers in the two test conditions. In addition, the NLPD controller outperforms the PD controller in terms of disturbance rejection.

\begin{remark}
In practice, the AUV parametric unknowns, including mass parameters, cannot be accurately modeled due to complex marine operations and disturbances. The idea behind the illustration of the differences for both Figs. \ref{fig-Rob1}--\ref{fig-angle_disturbance}, and Figs. \ref{fig-Rob2}--\ref{fig-angle_disturbance2} is to demonstrate the proposed NLPD controller scheme's ability to assess its substantial  robustness property to parametric uncertainties.
\end{remark}

\begin{table}[!t]
    \caption{Performance evaluation between PD and NLPD controllers with $20\%$ of mass parameter errors, ocean current disturbances and measurement noises---Case II.}\label{tab:rmse_para}
    \centering
  \begin{tabular}{cccccc}
\toprule										
Control scheme	&	$\mathrm{ RMSE}_x$[{\si m}]	&	$\mathrm{RMSE}_y$[{\si m}]	&	$\mathrm{RMSE}_\varrho$[{\si rad}]	&	~$t_a$[{\si s}]~	&	~$\Delta t$[s]~	\\	\midrule
PD controller	&	1.114	&	1.007	&	0.189	&	2.765	&	1.08	\\	
NLPD controller	&	0.696	&	0.601	&	0.189	&	3.040	&	--	\\	\midrule
Improvement	&	\textbf{37.52\%}	&	\textbf{40.32\%}	&	0\%	&	-0.0995	&	-	\\	\bottomrule
\end{tabular}
\end{table}

\section{Conclusion}\label{conclusion}
This paper proposes a solution to follow a mobile ship receiver by an AUV transmitter to guarantee a directed optical LoS link. The model of the LoS optical link between these two systems has been analyzed. Position tracking control strategies have been designed for the AUV to maintain a good position to the ship mobile receiver systems while satisfying a desired bit error rate.
Lyapunov function-based analysis that ensures the asymptotic stability of the resulting closed-loop tracking error is used to design the proposed NLPD controller. Through the simulation results in MATLAB/Simulink, the effectiveness of the proposed controllers to achieve favorable tracking in the presence of the solar background noise within competitive times is illustrated. The proposed NLPD based controller has been compared with the PD controller, and its performance is better than the PD controller. Additionally, numerical results demonstrate how the proposed NLPD controller improves the tracking error performance more than $70\%$ under nominal conditions and $35\%$ with model uncertainties and disturbances compared to the original PD strategy. It is worth noticing that the simulation results we have obtained provide a ``proof of concept" for short-range communication to access and cover a restricted subsea area through solar background noise.

Future work will focus on implementing the proposed controller's algorithms in scenarios-based real-time experimental and elevating the operational tracking and communication link budget performances in practices.

\bibliography{MZ-1}
\end{document}